\documentclass[11pt]{article}
\usepackage[utf8]{inputenc}
\usepackage[backend=bibtex,giveninits=true]{biblatex}
\usepackage{biblatex}
\usepackage{subfig}
\usepackage[normalem]{ulem}
\usepackage{float}
\usepackage{enumitem}

\bibliography{main}

\usepackage{setspace}
\usepackage{xcolor}

\usepackage[margin=1in]{geometry}
\usepackage{comment}

\usepackage{mathtools}

\usepackage{amsmath,amsthm,amssymb,latexsym,amsfonts, dsfont}
\usepackage{graphicx} 
\graphicspath{ {\string~/Desktop/} }
\usepackage[font=footnotesize,labelfont=bf]{caption}
\usepackage{mathabx}
\usepackage{hyperref}
\usepackage{lineno}
\usepackage{amsthm}
\usepackage{bm}
\usepackage{listings}
\usepackage[nottoc,numbib]{tocbibind}
\usepackage{xcolor,colortbl}

\usepackage{thmtools} 
\usepackage{thm-restate}
\usepackage{lipsum}

\newtheorem{theorem}{Theorem}

\newtheorem{case}{Case}
\newtheorem{claim}{Claim}

\newtheorem{lemma}{Lemma}

\newtheorem{proposition}{Proposition}

\newtheorem{observation}{Observation}
\newtheorem{fact}{Fact}

\newcommand{\eps}{\varepsilon}

\DeclareMathOperator*{\argmin}{arg\,min}
\DeclareMathOperator*{\argmax}{arg\,max}
\def\polylog{\operatorname{polylog}}

\allowdisplaybreaks

\DeclareCiteCommand{\citeyearpar}
    {}
    {\mkbibparens{\bibhyperref{\printdate}}}
    {\multicitedelim}
    {}

\title{Online $k$-Median with Consistent Clusters\footnote{Moseley and Newman are supported in part by  a Google Research Award, an Inform Research Award, a Carnegie Bosch Junior Faculty Chair, and NSF grants CCF-2121744 and  CCF-1845146. Pruhs is supported in part by NSF grants  CCF-1907673,  CCF-2036077, CCF-2209654 and an IBM Faculty Award.}}

\author{Benjamin Moseley\thanks{Tepper School of Business. Carnegie Mellon University.} \and Heather Newman\thanks{Department of Mathematical Sciences. Carnegie Mellon University.} \and Kirk Pruhs\thanks{Computer Science Department. University of Pittsburgh.}}

\begin{document}
\maketitle

\abstract{We consider the online $k$-median clustering problem in which $n$ points arrive online and must be irrevocably assigned to a cluster on arrival.  
As there are lower bound instances that  show that an online algorithm cannot achieve a competitive ratio that is a function of $n$ and $k$, we consider a beyond worst-case analysis model in which the algorithm is provided a priori with a predicted budget $B$ that upper bounds the optimal objective value. We give an algorithm that achieves a competitive ratio that is exponential in the the number $k$ of clusters, and show that the competitive ratio of every algorithm must be linear in $k$.
To the best of our knowledge this is the first
investigation in the literature that  
considers cluster consistency using competitive
analysis.
}

\section{Introduction}

Clustering problems, such as $k$-means clustering
and $k$-median clustering, are a 
classic genre of learning / data mining problems~\cite{bishop}. 
Typically the input consists of a collection $X=\{ x_1, \ldots, x_n\}$ 
of points in some metric space $\mathcal M$ (typically $\Re^d$ with the 1-norm or 2-norm)
and a positive integer $k$. 
The output for a \textbf{center-based} clustering problem is a collection
$c_1, \ldots, c_k$ of $k$
points from $X$, called centers, that succinctly summarize the data points.
The implicit cluster $C_i$
corresponding to the center $c_i$ is the collection of points in
$X$ whose closest center is $c_i$, that is $C_i =\{ x_j \mid \argmin_{h \in [k]} d(x_j, c_h) = i  \}$, where $d(\cdot, \cdot)$ is the distance function for the metric space.
The output for a \textbf{cluster-based} clustering problem is 
a partition $C_1, \ldots C_k$ of
$X$ into $k$ parts, called clusters.
The implicit center
of each cluster $C_i$ is then  $c_i = \argmin_{x_h \in C_i} \sum_{x_j \in C_i} d(x_h, x_j)$.
For both center-based clustering and cluster-based clustering, the objective is to minimize the cost
of the clustering.  This paper considers the $k$-median objective which is 
the aggregate  distance from each point to the center of its cluster, that is 
$\sum_{i=1}^k \sum_{x_j \in C_i} d(x_j, c_i)$.

Here we consider applications where the data points in $X$
arrive online over time. In a center-based clustering problem, the online algorithm
maintains a collection of centers.   In a cluster-based problem, 
the online algorithm needs to assign the data points  to a cluster
when they arrive, that is each point $x_j$ needs to be assigned 
a label $\ell_j \in [k]$ when $x_j$ arrives. 

An application of online clustering given by \cite{LSSTalk} is the task of
clustering news articles that arrive online.  For example at 
Yahoo news or Google news. We will refer to these outlets as the news providers. 
The  news provider selects some (approximately) fixed number $k$ of articles to feature
on the news homepage, and has a ``view complete coverage'' link 
next to each article to see all the news stories
on this topic. 
The problem of selecting  the best $k$ articles that summarize all current news articles is better modeled as a center-based clustering.
The  problem of partitioning all news  articles into clusters of similar articles is better modeled as a cluster-based clustering. 
Other applications can be
found in \cite{liberty2016algorithm,LattanziV17}, but we will use the
news story clustering application as our running canonical example.

A line of research~\cite{pmlr-v130-guo21a,FLN21,LattanziV17,liberty2016algorithm} within online clustering goes by moniker of consistent clustering.  Research on consistent clustering studies the trade-offs and relationship between the following objectives:
\begin{itemize}
    \item {\bf Maximizing the Quality of the Clustering:} One seeks
    the cost of the clustering to be small. 
    The most common metric to measure the quality of a
    solution is the ratio of the cost of this   solution to cost of the optimal solution.  
The most common metric to measure the quality of an online algorithm
is the competitive ratio, which is the maximum (over all inputs) of the ratio of the cost of
the online algorithm's solution to the optimal cost. 
    \item {\bf Maximizing Consistency:} Ideally one would like the centers in
    a center-based problem, or
   the clusters in a cluster-based problem, to be consistent over time.
 That is, they should change as little as possible. So for example,
the news provider doesn't want the clusters to completely change
every time a new news article is written.
\end{itemize}

\subsection{Prior Work on Consistent Clustering}

$k$-median clustering is 
NP-hard, but constant factor approximation polynomial-time algorithms
are known~\cite{CGST99,JV99,AGKMMP01,LS13,BPRST17}. 

All prior algorithmic research on consistent clustering that we are aware of~\cite{pmlr-v130-guo21a,FLN21,LattanziV17,liberty2016algorithm} is
center-based. That is, the  online algorithm  explicitly maintains a collection of
centers, and the clustering is implicit. That is, each point is assumed to be associated with the closest center, but there are no restrictions on how often points' associated centers can change.

 The first paper~\cite{liberty2016algorithm} in this line of research
 gave a lower bound that  
 showed that one can not simultaneously
 have both high quality and maximum consistency.
 That is, they showed that if a center cannot be changed once it is established, 
then there is no algorithm whose competitive ratio 
can be bounded by any function of $n$ and $k$.

Thus various ``beyond worst-case analysis'' (see~\cite{roughgarden_2021}) approaches have been used in the literature to 
attempt to circumvent
the obstacle presented by this lower bound. 
One approach is to use bi-criteria 
analysis or \emph{resource augmentation} analysis.  This analysis   allows the online algorithm to use more than $k$ centers,
and then compares the cost of the algorithm's clustering 
to the optimal one using $k$ 
centers~\cite{liberty2016algorithm}.
A second approach is to allow the algorithm \emph{recourse}, which 
in this setting means allowing the algorithm to  
change the centers (or clusters) a small number of times~\cite{LattanziV17,FLN21,pmlr-v130-guo21a}. 
Another approach is to consider learning augmented algorithms,
which assumes that the algorithm receives some advice/information a priori about the input.
For example, in the  news application, the news provider  presumably has prior data
that it could use to predict with some reasonable accuracy some properties of the input.

\cite{liberty2016algorithm}  gives a randomized algorithm for $k$-means clustering  and
analyzes this algorithm using \emph{resource augmentation} analysis. 
 \cite{liberty2016algorithm} shows that 
the expected number of clusters/centers used by their algorithm is 
$O( k \log n \log n \Delta )$ and at all times 
 the expected cost of the clustering using these centers
 is at most $O(\log n)$ times
the optimal cost using $k$ 
clusters. Here $\Delta$ is the aspect ratio of the data points,
which is the ratio between the distance between the furthest pair
of points and the distance between the closest pair of points.
The algorithm 
leverages 
a randomized online algorithm for facility location from \cite{Meyerson2001} 
to decide whether to create a new center at a newly arriving data point. 
Once a center is established, it is maintained throughout the course of the algorithm. 
Finally, \cite{liberty2016algorithm} gives a randomized algorithm
that requires a priori knowledge of
$n$ and a lower bound of the optimal with $k$ centers, 
and that maintains a  collection of $O(k \log n \log \alpha)$ centers in expectation
that has expected cost 
$O(1)$ times
the optimal cost with $k$
centers. Here $\alpha$ is the ratio between actual optimal cost with $k$
centers and the lower bound provided a priori to the algorithm.

\cite{LattanziV17} give  a randomized algorithm for $k$-means or $k$-median
clustering that uses \emph{recourse}.
The algorithm maintains the invariant that the cost of the current centers is
always $O(1)$-competitive with the optimal clustering of the data points seen \textit{to date}.
To maintain this invariant, the expected number of cluster center changes used by
the algorithm is 
 $O(k^2 \log^4 n \Delta)$. 
\cite{LattanziV17} show a similar lower bound, that is they show that every
algorithm requires 
 $\Omega(k \log_c  \frac{ \Delta}{k})$ center changes  to maintain $O(c)$-competitiveness.
Further, \cite{LattanziV17} show that it possible to maintain
$O(1)$-competitiveness with $O(k \log^2 n \Delta)$ center changes,
but this is just an existential result, and no algorithm to achieve
this was given. In followup paper \cite{FLN21} gave
a randomized algorithm that maintains $O(1)$-competitiveness with
$O(k \polylog( n \Delta)$ cluster changes. 
The results in \cite{LattanziV17} were extended to $k$-median
clustering with outliers (so one could opt to not
cluster a pre-specified number of points) in \cite{pmlr-v130-guo21a}.
\cite{LattanziV17}  also observes
that for $k$-center clustering an algorithm from \cite{CharikarCFM04} 
can be used to maintain
an $O(1)$-competitive clustering with $O(k \log  n \Delta)$ center changes. 

While not directly germane for the work in this paper, there is also research on online clustering in the streaming setting, where the emphasis
is more on the algorithm using a small amount of memory, or quickly responding to
the arrival of a new data point~(e.g. \cite{chanKDD,Cohen-AddadHPSS19}).

\subsection{Our Contribution}

Our initial research goal was to investigate consistent clustering for
cluster-based problems (recall that all the past algorithmic consistent
clustering publications that we are aware of focus on center-based clustering). 
We are interested in applications
where the focus is on explicitly maintaining consistent clusters (and not necessarily on maintaining consistent centers). 
The application where Google or Yahoo news is trying to maintain collections
of similar news articles is an example of such an application. 
Note that even the algorithms from \cite{liberty2016algorithm} that are
 perfectly consistent from a center perspective, in that once 
a center is established it persists until the end of the algorithm, 
are not necessarily consistent from a cluster perspective in that
a data point could  change clusters every time a  new center
is established. All one can say (at least naively) about 
the cluster consistency of the algorithms from \cite{liberty2016algorithm}
is that no data point changes clusters more than $O( k \log n \log n \Delta )$ times.

Specifically, our research goal is to determine whether some reasonable
competitiveness with respect to cost can be achieved if each data point has to be irrevocably 
assigned a cluster, or equivalently a label, when it arrives. 

The lower bound from \cite{liberty2016algorithm} implies  that
we will need to take some beyond worst-case analysis approach as otherwise there is no algorithm with bounded approximation ratio.
For reasons we explain in Section \ref{sect:overview},
we will take a learning augmented algorithms approach, and
assume that the algorithm is provided a priori with 
an estimated upper bound $B$ of the cost of the final optimal clustering.
Our algorithms will then try maintain
a clustering of low cost relative to $B$, \textit{not} the current optimal cost for
the points that have arrived to date. Thus we will say an algorithm
is $c$-competitive with respect to the budget if the algorithms cost
is at most $c \cdot B$ on instances where the optimal cost is at most $B$ (after all points arrive).

We  develop a lower bound in Section \ref{appendix: lower_bd}  that the competitiveness of any algorithm relative to the budget
must be $\Omega(k)$. In light of this observation, a natural question is whether
one can find  an algorithm whose competitiveness with respect to the budget is
$O(f(k))$ for some function $f$, that depends only on the number of centers $k$, and not on the number of points $n$.  Such an algorithm is constant competitive when $k$ is fixed.  We answer this in the affirmative by
giving a (polynomial-time deterministic) online  algorithm that is $O(k^5 3^k)$-competitive with respect to the budget.  Thus, we know the competitive ratio must depend on $k$ by the lower bound and we give an upper-bound that only depends on $k$.
Note that in most clustering applications while the number of points $n$ to be 
clustered may be large, the number $k$ of clusters is typically a modestly
small constant, say 10 for example~\cite{bishop}.

\section{Technical Overview}

\label{sect:overview}

To understand the motivation for learning-augmented approach, 
let us consider the lower bound instance from~\cite{liberty2016algorithm}.
It is sufficient to assume $k=2$. The first point $x_1$ arrives and is assigned some irrevocable label. Then assume the second data point $x_2$
arrives a unit distance from $x_1$. If the online algorithm assigns $x_2 $ the same label as  
$x_1$ then the cost of the algorithm's clustering is 1, and the optimal
cost is 0 (which is the cost if each of these data points were given a different labels). 
This results in the algorithm having in unbounded competitiveness. In contrast, if the algorithm gave
$x_2$ the a different label as $x_1$ then the third data point $x_3$ could arrive very far
away. In which case, the algorithm's clustering would necessarily have very high
cost (as $x_3$'s label would have to be either the same as $x_1$'s or the same as $x_2$'s). 
However, the optimal clustering would have cost 1 (by giving $x_1$ and $x_2$ the same label
and giving $x_3$ the remaining label).  Again, this results in competitiveness that can only be bounded by the aspect ratio of metric space (much larger than $n$ or $k$).

Intuitively, the dilemma faced by the online algorithm when $x_2$ arrives is
that it does not know whether the distance between $x_1$ and $x_2$ is small or large. 
Equipped with an estimate of the optimal cost $B$, the algorithm could 
 resolve this dilemma by giving $x_2$  a different label than $x_1$ if their distance is larger than $B$ and the same label otherwise. Thus we will assume
 that the algorithm is provided a priori 
 with a budge $B$ that is an upper bound on
 the optimal cost (and ideally is not too much larger than the optimal cost).

To better understand 
the design of our algorithm it is useful to understand some instances that illustrate some properties
that a competitive algorithm must have.

The first observation is that any reasonably competitive algorithm can 
never use $t+1$ labels if it is the case that the points to date could be
clustered with cost at most $B$ using at most $t$ labels.
If the algorithm ever allowed this to happen,  it could be that the
next $k-t$ data points could arrive very far from the previous data points,
and very far from each other. Thus after these data points arrive,
the algorithm's cost would
increase to something like the diameter of the metric space, while there
would still be a clustering of cost at most $B$, since the clustering that
used $t$ labels could be extended with no additional cost by giving each
new $k-t$ data points a new label.

In light of this last observation, a natural greedy algorithm would maintain the
invariant that the number of labels it uses is always equal to the minimum
number of labels necessary for a clustering of cost at most $B$,
and then give each new data point the label that minimizes the increase in cost.
To see why such an algorithm (and other similar algorithms) can have unbounded cost even when $k=2$, 
and the metric space is the real line,
consider the following instance (see Figure \ref{fig: greedy_tricky_ex}). 
Let $\alpha$ be an arbitrarily large positive integer. We construct an example in which the budget $B=2$.  
The first point arrives  at location $-2$. Say the algorithm gives this point 
the label blue. The next point arrives at location $1$. Now, we know that any offline clustering with cost at most $2$ must use at least 2 clusters. So the greedy algorithm would give
this new point a second label, say green, as that would minimize the increase in the objective. Then a total of $\alpha$ additional points arrive at location $1$. The algorithm labels these points green. Then $\alpha$ points arrive at the origin 0. It is still the case that only 2 clusters are required in order  to have cost at most $2$, since we may place the points at location $-2$ and the origin in one cluster, and
the  points at location 1 in the other cluster.  However the algorithm would assign each point arriving at the origin the label green, since this increases the objective by at most $1$ while assigning such a point the label blue increases the objective by $2$.   Yet, this results in a solution for the algorithm in which the contribution of green points towards the objective is $\alpha $. 

\begin{figure}[H]
\centering
\includegraphics[width=0.8\textwidth]{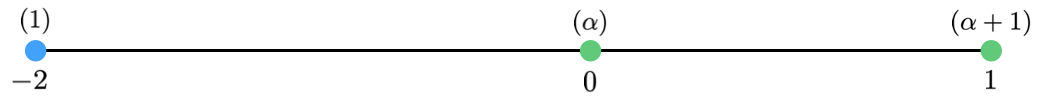}
\caption{An example in which the natural greedy algorithm incurs arbitrarily large cost.}
\label{fig: greedy_tricky_ex}
\end{figure}

Upon reflection of this lower bound instance for the natural greedy algorithm,
there appear to us to be two natural hypotheses as to the ``mistake'' that this algorithm is making,
and corresponding two natural paths towards a remedy:
\begin{itemize}
    \item One hypothesis is that greedy assignment is a mistake,
    and then the natural remedy would be to use some label assignment rule that is more sophisticated than greedy.
    \item
    Another hypothesis is that the algorithm was too hasty in using a new label.
    Thus the natural remedy  would then be to delay using a new label  until it is more clear as to
a region where arriving data points should be given this new label. 
Note in the example in Figure \ref{fig: greedy_tricky_ex} that if the algorithm had waited until some reasonable
number of data points had arrived at the origin before using the second label, then the algorithm might reasonably have been able to 
see that the right choice was to give the  remaining points arriving at the origin the second label of green.
\end{itemize}

Here we primarily adopt the second remedy/approach (although
we also consider an alternate greedy assignment policy). 
To apply this remedy we must address the following two questions:\begin{itemize}
    \item 
 Under what conditions can the algorithm
justify the use of an additional label, say going from $t-1$ labels to $t$ labels?
\item
And when this can be justified, how should we modify our prior partition of space into $t-1$ parts into a partition into $t$ parts?
\end{itemize}
At a high level our answer to the first question is that we do not use $t$ labels until there  exist $t$ well-separated points $x_{\alpha(1)}, \ldots, x_{\alpha(t)}$. 
We will say that a collection of points $x_{\alpha(1)}, \ldots, x_{\alpha(t)}$ from a collection $S$ of points is $\beta$-\textbf{well-separated} with respect to $w_S$ if
\begin{equation} \label{eq: well-sep-def}
\min\{w_S(x_{\alpha(i)}),w_S(x_{\alpha(j)})\}\cdot d(x_{\alpha(i)},x_{\alpha(j)}) \geq \beta \cdot B \hspace{0.3cm} \text{for all } i,j \in [t], i \neq j \tag{$\star$}
\end{equation}
Here $w_S(x_h)$ is what we call the \textbf{natural weight} of point $x_h$ in $S$, which
is the maximum number of points in $S$ whose distances to $x_h$ sum to at most $2B$. (If $S$ is clear, we may drop it from the notation).
Intuitively, if we have $t$ well-separated points then 
we then know that not only must any near optimal solution use $t$ labels, but such a solution cannot combine
the points near $x_{\alpha(i)}$ and the points near $x_{\alpha(j)}$ into a single cluster (assuming $i \ne j$).

To drill down a bit further, we maintain a collection of points $p_1, \ldots, p_t$ from the online stream $X$ which we call \textit{pivots}.
When a new point arrives, it is assigned the label $i$ of the pivot $p_i$ nearest to it (so we still maintain a form of greedy label assignment). 
While one might reasonably think that the pivots are intuitively centers for the clusters, 
this intuition is only partially correct. In fact there are scenarios where some pivots are
in fact poor centers for the corresponding cluster. What is critical is that the pivots are located
so as to guarantee that using greedy assignment in the future results in a relatively low cost assignment. 
In order for our cost analysis to be tractable  we want the algorithm to maintain the following invariants:\begin{itemize}
    \item Each pivot $p_i$  is reasonably located in a region where it would not be too costly to assign points arriving in this region the label $i$.   
    \item The pivots are well-separated (for some appropriate choice of $\beta$).\footnote{$\beta$ will have to both be initialized sufficiently large and also decrease as the number of pivots increases. We show this is necessary in Appendix \ref{appendix: decreasing_well_sep}.}
    \item There is no other point that is well-separated from the pivots.
    \item The locations of the pivots should not move very often.
\end{itemize}
Note that some of these invariants can intuitively be in opposition to each other, which seemingly requires
that the design of the algorithm
is a bit complicated, as there are several different scenarios where
maintaining this invariant requires different actions. 
But we will now try to give a brief, definitely over-simplified, overview of how the algorithm maintains these invariants.
As the pivots are not necessarily good centers (for example pivot $p_1$ at location $-2$ as the points arrive at location 1 in Figure \ref{fig: greedy_tricky_ex}), the algorithm also maintains a collection
$c_1, \ldots, c_t$ of estimated centers for the $t$ labels that have been used to date. 
The pivots and the estimated centers are updated in two scenarios. The first scenario is when there is an applicable Add Operation, and the second
is when there is an applicable Exchange Operation. 
In both scenarios the estimated centers are recomputed from the previously computed estimated centers and the newly computed (near) optimal centers. 
Then a sequence of Add and Exchange Operations are executed. 

An Add Operation is applicable when there is a point $x_\alpha$ that is well-separated from
the current pivots. Intuitively, this means a new label can be justified, but the implementation requires the consideration of several possible scenarios. 
 In the simplest scenario $x_\alpha$ is near a cluster of new points that are all far from
previous points, and the pivot $p_{t+1}$ for the new label ($t+1$) is set to $x_\alpha$. In some scenarios  an old pivot $p_i$ ($i \le t$) is set to $x_\alpha$ and $p_{t+1}$ is set to $p_i$ (so the new pivot 
inherits an old label and an old pivot location gets the new label). Intuitively, this occurs when the estimated center $c_i$ for cluster $i$ is near the location of $x_{\alpha}$.
Also there are scenarios where an old pivot $p_i$ ($i \le t$) is set to the estimated center $c_i$ of cluster $i$ 
and $p_{t+1}$ is set to $p_i$ (so again the new pivot inherits an old label and an old pivot location gets the new label). 
Finally, there are scenarios where two new pivots are created.\footnote{This is needed to avoid label conflicts. See Appendix \ref{appendix: label_conflicts}.}

An Exchange Operation is applicable when there are two  points $x_\alpha$ and $x_\gamma$ 
 near a pivot $p_j$ that are well-separated from each other and 
the other pivots (besides $p_j$). So intuitively the cluster of points labeled $j$ appear to be splitting
into two clusters. In the simplest scenario the location of pivot $p_j$ is set to the location of one of
$x_\alpha$ or $x_\gamma$, and the location of the new pivot $p_{t+1}$ is set to the location of the other one. 
This scenario occurs in the instance depicted in Figure \ref{fig: greedy_tricky_ex}.
The first pivot $p_1$ is initially set to location $-2$. The points arriving at
location $1$ would all be assigned the label 1 (blue) as there is no  point well-separated
from $p_1$ (the points located at $1$ are not separated from $p_1$ because the  points at
$p_1$ can be cheaply moved to location $1$). When enough points have arrived at the origin,
then the points $x_\alpha = 0$  and  at $x_\gamma = 1$ are near $p_1$ (because the point at 
$p_1$ can be cheaply moved to either $x_\alpha$ or $x_\gamma$), and are well-separated from each other
and the pivots other than $p_1$. Thus our algorithm would locate $p_1$ at $1$ and $p_2$ at the origin.
While this scenario gives some intuition,
there are several other more complicated scenarios.

The analysis is broken into two parts. The first part is to show that certain stuctural invariants are maintained through carefully engineered algorithmic choices.  The well-separated invariant is used to show we do not use more than $k$ labels and for bounding the cost. The second part of the analysis is to bound the cost (Section \ref{sec:cost}). We inductively assume that the cost of points that were labelled using the first $t-1$ labels is bounded. Then, we show that the points that were labelled when exactly $t$ labels were in use are clustered near optimally under our greedy procedure. The key challenge is showing the points given the same label combined have bounded cost.

\section{Preliminaries}
First we establish a lower bound depending on $k$ for the competitive ratio of our algorithm. 

\begin{theorem} \label{thm: lower_bd}
The competitiveness versus the budget of any deterministic algorithm
for cluster-centric consistent clustering is $\Omega(k)$. 
\end{theorem}
The proof of Theorem \ref{thm: lower_bd} is in Appendix \ref{appendix: lower_bd}. 

\textbf{Assumptions.} We only assume that the online stream $X$ lies in a metric space. We allow multiple points to arrive at the same location; each duplicated point still has to pay the cost to its nearest center. Whenever we refer to centers, we enforce these come from $X$ itself. (This will be used to ensure our algorithm is well-defined.)

We state our results assuming $B = \textsf{OPT}$, but all results still hold by replacing $\textsf{OPT}$ with $B$, as long as $\textsf{OPT} \leq B$. 
Our algorithm will need to compute offline $k$-median clusterings; to do this in poly-time, use any $c$-approximation algorithm. Thus, replace $B$ with $c \cdot B$ to run our online algorithm in poly-time. 

\textbf{Terminology.} Recall the term \textit{natural weight} from Section \ref{sect:overview}. We will always take $S$ to be some prefix of $X$. As such, we may refer to the natural weights at a particular point in time to mean the natural weights in $S$, where $S$ is the prefix at that time.

Note that for $p \in S$, $w_S(p) \geq 1$, since $d(p,p) = 0$, and that $w_S(p)$ can only increase over time as $S$ enlarges. 

Recall the term \textit{$\beta$-well-separated} from Section \ref{sect:overview}. Related terms we will use are the following: We will say $p$ is $\beta$-\textbf{well-separated from} a set of points $\{x_{\alpha(1)}, \dots, x_{\alpha(m)} \}$ w.r.t. $w$ if $\min\{w(p), w(x_{\alpha(i)})\}\cdot d(p,x_{\alpha(i)}) \geq \beta \cdot \textsf{OPT}$ for all $i \in [m]$. If $m=2$ and the well-separated condition (\ref{eq: well-sep-def}) is not satisfied, we say the pair of points is $\beta$-\textbf{attached}, or that one is $\beta$-attached to the other, w.r.t. $w$.

\section{Algorithm Description} \label{sec: algo}
The algorithm sees an online sequence $X= \{x_1, x_2, \ldots x_n\}$ of points.
Let $X_i =  \{x_1, x_2, \ldots x_i\}$ and $w_i$ be shorthand for $w_{X_i}$. 
The algorithm maintains the following:
\begin{itemize}
    \item a collection of previously arriving points  $p_1, \dots, p_t$ that  have been designated as \textit{pivots},
where $t$ is the number of labels used by the algorithm to date and pivot $p_j$ is associated with label $j$, 
\item
a separation parameter $\beta_t = 8 \cdot 3^{k-t+2}$, and
\item
a collection of previously arriving points $c_1, \ldots, c_T$ that have been designated as estimated centers, where $T \le t$.
\end{itemize}
Initially the first point $x_1$ is given the label 1, the first pivot $p_1$ is set to $x_1$,
and the collection of estimated centers is empty.
The algorithm handles the arrival of each subsequent point $x_i$ in the following manner. The subroutines are described in Sections \ref{sec: est_center_subroutine}, \ref{sec: add_op}, \ref{sec: exchange_op}. 
\begin{enumerate}[label={(\arabic*)}]
    \item \textbf{If} there is an applicable Add or Exchange Operation \textbf{then} compute new Estimated Centers
\begin{enumerate} [label=(\alph*)]
        \item \textbf{Repeat} \text{ while there is an applicable Add Operation or Exchange Operation} 
        \begin{enumerate}
            \item \textbf{If} \text{ there is an applicable Add Operation}
            \textbf{ then} \text{ apply an arbitrary applicable one} 
            \item \textbf{Else} \text{ apply an arbitrary applicable Exchange Operation } 
        \end{enumerate} 
    \end{enumerate}
    \item Give $x_i$ the label $j$, where $p_j$ is the nearest pivot to $x_i$
\end{enumerate}

Let $T$ be the number of pivots during an execution of the outer loop (1). The Estimated Centers subroutine computes $T$ new estimated centers
$c_1, \ldots, c_T$ from the current pivots $p_1, \ldots p_T$, 
the points $X_{i-1}$, and the current estimated centers $c_1, \ldots c_s$ ($s<T$). 

Let $t \geq T$ be the number of pivots during an execution of the inner loop (a). The Add Operation subroutine is applicable if there is a point $x_\alpha \in X_i$ such that $x_\alpha$ is $\beta_{t+1}$-well-separated from 
the current pivots $p_1, \dots, p_t$ with respect to the weight function $w_i$.
The execution of the Add Operation subroutine depends on $x_\alpha$, $X_i$, the current pivots $p_1, \ldots, p_t$,
and the current estimated centers $c_1, \ldots, c_T$. The Add Operation subroutine adds one or two new pivots, and possibly changes the location of up to two previous pivots.

The Exchange Operation subroutine is applicable  if there
there exists two points $x_\alpha$ and $x_\gamma$ in $X_i$, and a pivot $p_j$ such that:
\begin{itemize}
    \item $x_\alpha$ and $x_\gamma$ are each $\beta_{t+1}$-attached to $p_{j}$ w.r.t. $w_{i}$,
    \item $w_i(p_{j}) \leq w_i(x_\alpha)$,
        \item $w_i(p_{j}) \leq w_i(x_\gamma)$, and
\item The collection of the $t+1$ points, consisting of $x_\alpha$, $x_\gamma$, and the pivots other than $p_j$, are
 $\beta_{t+1}$-well-separated w.r.t. $w_i$.
\end{itemize}
The execution of the Exchange Operation subroutine depends on $x_\alpha$, $x_\gamma$, $X_i$, the current pivots $p_1, \ldots, p_t$,
and the current estimated centers $c_1, \ldots, c_T$. 
The Exchange Operation adds one or two new pivots, and possibly changes the location of one previous pivot.

\subsection{The Estimated Center  Subroutine} \label{sec: est_center_subroutine}

Choose $y_1, \ldots, y_k \in X_{i-1}$ to be an optimal collection of $k$ centers for the points in $X_{i-1}$. (By Fact \ref{fact: Meyerson}, there exists a collection that gives a clustering of cost at most $2\textsf{OPT}$.) Define $p(y_h)$ to be the pivot with the minimum weighted distance to $y_h$, that is, 
\begin{equation} \label{eq: p_arc}
p(y_h) = \argmin_{p_j}  \left(\min\{w_{i-1}(p_j), w_{i-1}(y_h)\} \cdot d(p_j, y_h)\right) \tag{$\dagger$}
\end{equation}
For a pivot $p_j$, let $\delta(p_j)$ be a collection of optimal centers and estimated centers defined 
as follows: the current estimated center $c_j$ (when it exists) is in $\delta(p_j)$ if $w_{i-1}(c_j) > w_{i-1}(p_j)$ and $c_j$ is $\beta_{t+1}$-attached to $p_j$ w.r.t. $w_{i-1}$, and $\delta(p_j)$ contains an optimal center $y_h$ if $w_{i-1}(y_h) > w_{i-1}(p(y_h))$ for $p(y_h) = p_j$.
For each $j \in [T]$  we define the  new \textbf{estimated center} $c_j$ as follows:
If $w_{i-1}(p_j) \ge  \max_{p \in \delta(p_j) } w_{i-1}(p)$ then $c_j = p_j$, else
\begin{equation} \label{eq: good_center}
\centering
c_j = \argmax_{p \in \delta(p_j) } w_{i-1}(p) \tag{$\ddagger$}
\end{equation}

\begin{figure}[ht]
\centering
\includegraphics[width=.9\textwidth]{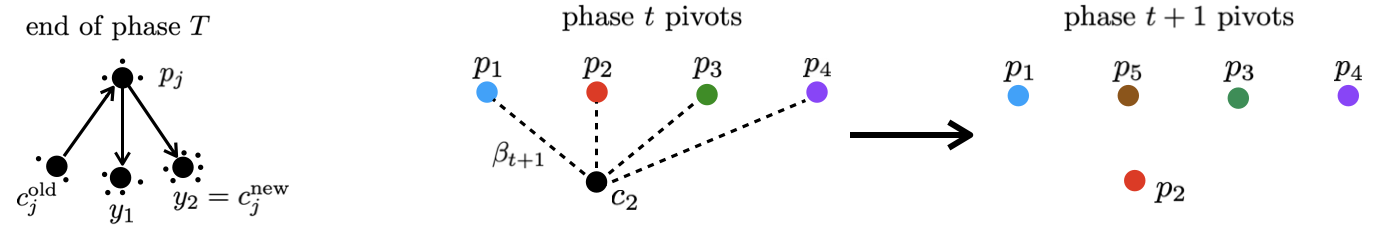}
\captionsetup{width=.9\linewidth}
\caption{\textbf{Left:} The Estimated Center subroutine. Arrows point from smaller to larger natural weights, and the small points around the larger points are the the set attaining the weight of the larger point. \textbf{Right:} Case (1) of the Add Operation. Dashed lines represent well-separation. Colors except black represent labels. The pivot for label 2 (red) moves to location $c_2$, while the old location for red's pivot becomes the pivot of the new label (5, brown).}
\label{fig: estimates_add1}
\end{figure}

\subsection{The Add Operation Subroutine} \label{sec: add_op}

For shorthand, let $w_t = w_{i-1}$ when $t=T$ and $w_t = w_i$ when $t > T$.\footnote{We are overloading subscripts here for ease. We could instead write $v_t$, but we retain $w$ to recall weights.}
 The description of the Add Operation subroutine is then: 
 
\begin{enumerate} [label={(\arabic*)}]
        \item \textbf{If}  there is an estimated center  $c_{j}$ that is  $\beta_{t+1}$-well-separated from $p_1, \dots, p_t$ w.r.t. $w_i$ then set $p_{t+1} = p_{j}$ and  set $p_{j} = c_{j}$.
        \item \textbf{Else if} it is the case that for every estimated center $c_j$ that is $\beta_{t+2}$-attached to $x_\alpha$ w.r.t. $w_i$ it is also the case that $w_t(c_j) < w_t(p_j)$, then set $p_{t+1} = x_\alpha$.
        \item \textbf{Else if} there exists a unique estimated center $c_{j}$ that both is $\beta_{t+2}$-attached to $x_\alpha$ w.r.t. $w_i$ and satisfies $w_t(c_{j}) \geq w_t(p_{j})$, then set $p_{t+1} = p_{j}$ and $p_{j} = x_\alpha$.
        \item \textbf{Else}  Let $c_{f}$  and  $c_{g}$ be estimated centers such that each is $\beta_{t+2}$-attached to $x_\alpha$ w.r.t. $w_i$,  $w_t(c_{f}) \geq w_t(p_{f})$, and $w_t(c_{g}) \geq w_t(p_{g})$. Then set $p_{t+1} = p_{f}$, $p_{t+2} = p_{g}$, $p_{f} = c_{f}$, and $p_{g} = c_{g}$. 
\end{enumerate}

\begin{figure}[H]
\centering
\includegraphics[width=.9\textwidth]{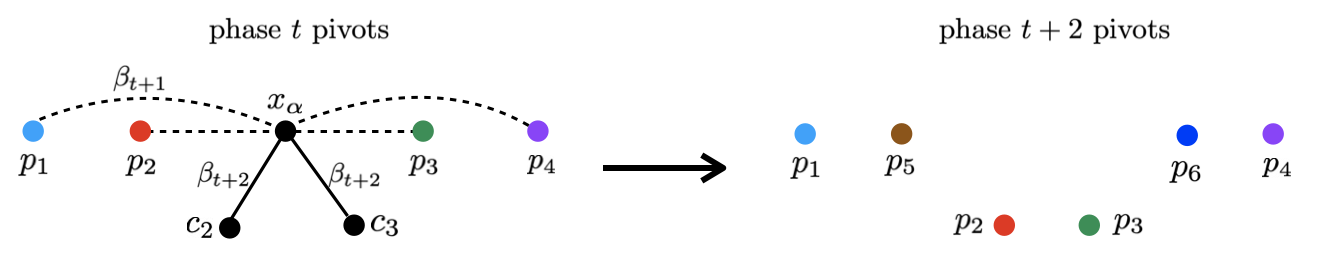}
\captionsetup{width=.9\linewidth}
\caption{Case (4) of the Add Operation. Dashed lines represent well-separation while solid lines represent attachment, both w.r.t. $w_i$ (labelled with the appropriate $\beta$). The pivots for labels 2 (red) and 3 (green) move to $c_2, c_3$, resp., while the old locations for these pivots are where the new pivots 5 and 6 (brown and blue) are located.}
\label{fig: add4}
\end{figure}

\subsection{The Exchange Operation Subroutine} \label{sec: exchange_op}

 The description of the Exchange Operation subroutine is: 
    \begin{enumerate} [label={(\arabic*)}]
        \item \textbf{If} $j > T$ then set $p_{j} = x_\alpha$ and $p_{t+1} = x_\gamma$.
            \item \textbf{Else if} $w_i(c_{j}) < w_i(p_{j})$ then set $p_{j} = x_\alpha$ and $p_{t+1} = x_\gamma$. 
            \item \textbf{Else if} $c_{j}$ is $\beta_{t+2}$-attached to $x_\alpha$ w.r.t. $w_i$ then  set $p_{j} = x_\alpha$ and $p_{t+1} = x_\gamma$.
            \item \textbf{Else if} $c_{j}$ is $\beta_{t+2}$-attached to $x_\gamma$ w.r.t. $w_i$ then set $p_{j} = x_\gamma$ and $p_{t+1} = x_\alpha$.
            \item \textbf{Else} set $p_{t+1} = x_\alpha$, $p_{t+2} = x_\gamma$, and $p_{j} = c_{j}$.
    \end{enumerate}

\begin{figure}[ht]
\centering
\includegraphics[width=.9\textwidth]{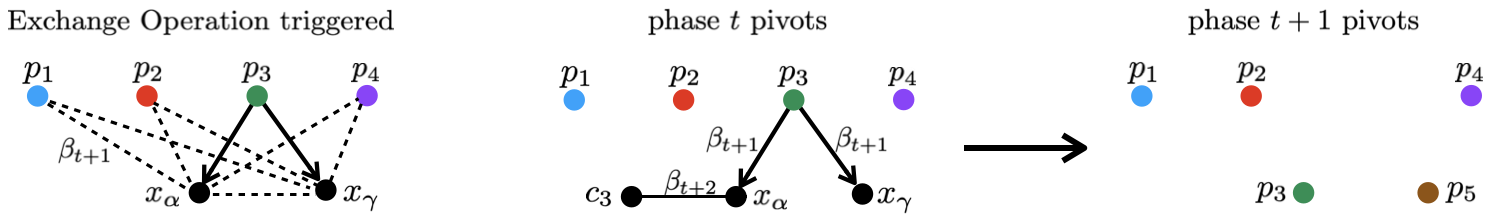}
\captionsetup{width=.9\linewidth}
\caption{Dashed and solid lines and arrows are as in Figures \ref{fig: estimates_add1} and \ref{fig: add4}. Colors except black represent labels. \textbf{Left:} The configuration triggering the Exchange Operation; dashed lines between $p_1,\dots,p_4$ suppressed. \textbf{Right:} Case (2) of the Exchange Operation. The old location for label 3's (green's) pivot is no longer a pivot location, and label 3's pivot is now $x_{\alpha}$. The new label, label 5 (brown), has its pivot at $x_{\gamma}$.}
\label{fig: exchange}
\end{figure}

\section{Algorithm Guarantees and Invariants}

We will establish two main guarantees for our algorithm and a set of invariants the algorithm will maintain to establish these guarantees. The most challenging and interesting algorithmic property is that the algorithm has bounded cost.  The second is that the algorithm is feasible. The theorem below gives the cost guarantees; its proof is in Section~\ref{sec:cost}. We state our results assuming $B = \textsf{OPT}$, but all results still hold by replacing $\textsf{OPT}$ with $B$, as long as $\textsf{OPT} \leq B$. 

\begin{theorem} \label{thm: main_thm}
The cost of the algorithm is $O(k^5\cdot 3^k \cdot \textsf{OPT})$. 
\end{theorem}

The next theorem states that the algorithm never uses more than $k$ labels and therefore produces a feasible solution. The proof of this theorem is in Section~\ref{sec:invariant}.

\begin{theorem} \label{thm: under_k_labels}
The algorithm uses at most $k$ labels. 
\end{theorem}

\subsection{Notation and Definitions} \label{sec: analysis_defs}
We now define some notation and definitions used in the analysis. 
\begin{itemize}
    \item \textbf{Phase $t$} refers to the set of time steps during which there are exactly $t$ pivots. 
    \item $p_1^t, \dots, p_t^t$ denote the pivots for labels 1 through $t$, respectively, during phase $t$.
    \item $w^t$ denotes the natural weights at the end of phase $t$. 
    \item $X(t)$ is the set of points assigned a label before or during phase $t$.  
    \item For $j \in [t]$, let $C_j^t$ denote denote the points labelled $j$ in phases 1 through $t$. 
    \item An \textbf{intermediate phase} $t$ is a phase during which no points are given labels. This is a phase solely used to reset pivots. 
    \item A \textbf{non-intermediate phase} $T$ is a phase in which at least one point is given a label. 
\end{itemize}

For a non-intermediate phase $T$, 
    \begin{itemize}
        \item Let $T^-$ to denote the most recent non-intermediate phase before $T$, and $T^+$ to denote the first non-intermediate phase after $T$ (when these exist).\footnote{Using the notation in Section \ref{sec: algo}, $x_{i-1}$ is the last point labelled during phase $T$, and $x_i$ is the first point labelled during phase $T^+$.}
        \item For $j \in [T]$, $c_j^T$ is the estimated center (\ref{eq: good_center}) computed at the end of phase $T$. 
        \item Let $y_1, \dots, y_k$ be the optimal collection of $k$ centers computed at the end of phase $T$ in The Estimated Center Subroutine. Let $P_T = \{p_1^T, \dots, p_T^T, y_1, \dots, y_k\}$ and call this set the \textbf{offline centers} for phase $T$.\footnote{Note that $c_j^T$ and $P_T$ are only defined for non-intermediate phases $T$.}
        \item The \textbf{attachment digraph} $D(T)$ is a bipartite  digraph with vertex set $P_T$, plus $c_j^{T^-}$ if $T>1$, partitioned as $(\{p_1^T, \dots, p_T^T\}, \{y_1, \dots, y_k, c_j^{T^-}\})$. There is a directed arc $(y_i, p(y_i))$ if $w^T(y_i) \leq w^T(p(y_i))$ and a directed arc $(y_i, p(y_i))$ otherwise. If $c_j^{T^-}$ and $p_j^T$ are $\beta_{T+1}$-attached w.r.t. $w^T$, add the arc $(c_j^{T^-}, p_j^T)$ if $w^T(c_j^{T^-}) \leq w(p_j^T)$ and the arc $(p_j^T, c_j^{T^-})$ otherwise. $\delta^+(p_j^T)$ and $\delta^-(p_j^T)$ denote the out-degree and in-degree of $p_j^T$, respectively.
\end{itemize}

\subsection{Invariants} \label{sec:invariants}

The first property of our algorithm is that it maintains pivots that are sufficiently far apart with respect to their natural weights. This is at the heart of our analysis for both controlling the number of labels used and the algorithm's cost.

\begin{lemma} \label{lem: well-sep-invariant}
Let $t \in [k]$. The algorithm maintains the invariant that $p_1^t, \dots, p_t^t$ are $\beta_t$-well-separated w.r.t. the natural weights at the start of phase $t$ (and thereafter).\footnote{Two points well-separated at one time step will also be well-separated at a later time step, since their natural weights can only increase and the well-separation parameter $\beta_t$ can only decrease.}
\end{lemma}

The next lemma is a key technical lemma. It states that the estimated center (\ref{eq: good_center}) for the points given label $j$ \textit{before} phase $T$ is close, in a weighted sense, to the pivot for label $j$ in phase $T$. This is key to showing that points in cluster $j$ that are labelled \textit{before} phase $T$ can be combined with those that are labelled \textit{during} phase $T$ at bounded cost.  This lemma is in tension with the prior lemma because a pivot must be placed in a location where it is both well-separated from other pivots and is close to center of mass of the points of a given label.  

\begin{lemma} \label{lem: well_attached_centers}
Let $T$ be a non-intermediate phase and let $j \in [T]$. Let $w^t$ denote the natural weights at the end of phase $t$. If $T>1$, then at least one of the following holds:
\begin{enumerate}[label=(\alph*)]
    \item  $w^{T^-}(c_j^{T^-}) \leq w^T(p_j^T)$ and $w^{T^-}(c_j^{T^-}) \cdot d(c_j^{T^-}, p_j^T) \leq \beta_{T^-}(T-T^-) \cdot \textsf{OPT}$, \underline{or}
    \item $c_j^{T^-}$ is $\beta_{T+1}$-attached to $p_j^T$ w.r.t. $w^T$. 
\end{enumerate}
\end{lemma}

\subsection{Proof of the Algorithm's Feasibility and the Invariants}
\label{sec:invariant}

We begin by showing a bound on the number of well-separated points in the entire point set.    Lemma \ref{lem: well-sep-invariant} along with Proposition~\ref{prop: less_than_k_centers} below will immediately imply Theorem \ref{thm: under_k_labels}, which states that the algorithm uses at most $k$ labels. 

\begin{proposition}\label{prop: less_than_k_centers}
Let $X$ be a set of points whose optimal $k$-median cost using $k$ centers is $\textsf{OPT}$. Let $\{x_1, \dots, x_l\}$ be a set of points in $X$, and let $w_X$ denote their natural weights in $X$. Let $\beta > 8$. If $\{x_1, \dots, x_l\}$ is $\beta$-well-separated w.r.t. $w_X$, then $l \leq k$.
\end{proposition}

\begin{proof}[Proof of Proposition \ref{prop: less_than_k_centers}]
For shorthand, let $w_i = w_X(x_i)$. By Markov's inequality, for each $i \in [l]$ there must be at least $w_i/2$ points from $X$ inside $B(x_i, 2\textsf{OPT}/w_i)$. Consider a clustering on $X$ with cost $\textsf{OPT}$ using $k$ centers. Then each ball $B(x_i, 4\textsf{OPT}/w_i)$ must contain at least one of these $k$ centers; for, if not, then at least $w_i/2$ points inside $B(x_i, 2\textsf{OPT}/w_i)$ must each pay strictly more than $2\textsf{OPT}/w_i$ to reach a center. This contradicts that the optimal $k$-median cost is $\textsf{OPT}$. 

Now it remains to show, using the well-separation assumption, that these balls are disjoint; this will imply that $l \leq k$, since each ball must contain a center. Suppose to the contrary that there exist $i,j \in [l]$, $i \neq j$, such that $p \in B(x_i, 4\textsf{OPT}/w_i) \cap B(x_j, 4\textsf{OPT}/w_j)$. Applying the triangle inequality gives 
\begin{align*}
d(x_i, x_j) \leq d(p,x_i) + d(p,x_j) &\leq \frac{4\textsf{OPT}}{w_i} + \frac{4\textsf{OPT}}{w_j} \leq \frac{8\textsf{OPT}}{\min\{w_i, w_j\}}
\end{align*}
which contradicts that $\{(x_i, w_i), (x_j, w_j)\}$ is $\beta$-well-separated, i.e., that $\min\{w_i, w_j\} \cdot d(x_i,x_j) \geq \beta\textsf{OPT}$, since $\beta > 8$. 
\end{proof}
 
\medskip

The next two propositions will be used to aid the proofs of Lemmas \ref{lem: well-sep-invariant} and \ref{lem: well_attached_centers}. Recall that for each non-intermediate phase $T$, we defined a set of offline centers $P_T$ that has cost at most $2\textsf{OPT}$ on $X(T)$ (Section \ref{sec: analysis_defs}). In order to compare the (low-cost) offline clustering induced by $P_T$ to our online algorithm's clustering, we relate the offline set of centers $P_T$ (which we \textit{know} have bounded cost on $X(T)$) to the pivots in phase $T$ (which are used to make the greedy online choices) in the next proposition. 

\medskip

\begin{proposition} \label{prop: digraph_attachment}
Let $P_T = \{p_1^T, \dots, p_T^T, y_1, \dots, y_k\}$ be as in Section \ref{sec: analysis_defs}. Then $y_i$ and $p(y_i)$ are $\beta_{T+1}$-attached w.r.t. the natural weights $w^T$ at the end of phase $T$. 
\end{proposition}

\begin{proof}[Proof of Proposition \ref{prop: digraph_attachment}]
First we show that $y_i$ is $\beta_{T+1}$-attached to at least one of $p_1^T, \cdots, p_T^T$ w.r.t. $w^T$. Suppose not. Then an Add Operation would have been executed instead, and phase $T$ would have terminated in the previous time step, which is a contradiction. That $y_i$ is $\beta_{T+1}$-attached to $p(y_i)$ in particular follows directly from the definition (\ref{eq: p_arc}) of $p(y_i)$ as the pivot with the minimum weighted distance to $y_i$. 
\end{proof}

\medskip

 Observe that the proof of Proposition \ref{prop: digraph_attachment} is a direct consequence of the fact that we always execute an Add Operation or an Exchange Operation when one is available, so during a phase none are available. Each attached pair in Proposition \ref{prop: digraph_attachment} is encoded in the digraph $D(T)$ by a directed arc. Due to Proposition \ref{prop: digraph_attachment}, we can now think of this directed arc as representing the direction in which we could move a certain number of points sitting near one endpoint to the other endpoint at bounded cost.

 Algorithmically, we used $P_T$ to define the estimated centers $c_j^T$ (\ref{eq: good_center}). Next we show that the estimated center for a cluster at the end of a phase is attached to the pivot for that cluster in that phase. Thus, while the pivot itself may not be a good center for the cluster, the pivot is close to the estimated center (in at least one direction, in a weighted sense). 

 \medskip

\begin{proposition} \label{prop: attached_estimated_center}
    The estimated center $c_j^T$ is $\beta_{T+1}$-attached to $p_j^T$ w.r.t. the natural weights $w^T$ at the end of phase $T$. Further, $w^T(c_j^T) \geq w^T(p_j^T)$, with equality if and only if $c_j^T = p_j^T$. 
\end{proposition}

The proof of Proposition \ref{prop: attached_estimated_center} is a straightforward consequence of Proposition \ref{prop: digraph_attachment} and the definition (\ref{eq: good_center}) of estimated center; for completeness, it can be found in Appendix \ref{appendix: omitted_proofs_invariant}.

Using these propositions, we establish Lemma~\ref{lem: well-sep-invariant}, used heavily in our analysis. The full proof is involved, so we defer it to Appendix \ref{appendix: omitted_proofs_invariant} and provide a sketch of the key ideas here. 

\medskip

\begin{proof}[Proof sketch of Lemma \ref{lem: well-sep-invariant}.]
The proof is by induction. However, we need to couple the induction with a statement about the relative position of the estimated center for a cluster (which stays fixed between intermediate phases) to that cluster's pivot, which may change often as we consecutively reset the pivots between intermediate phases. Roughly, we prove below that if the estimated center for cluster $j$ has not separated entirely from the present set of pivots, then it must be close (in a weighted sense) to the present pivot for label $j$. 
\begin{restatable}{proposition}{innerinduct}\label{prop: intermediate_attachment}
Let $w_{i-1}$, $w_i$, and $w_t$ be as Section \ref{sec: add_op}. For each $j \in [T]$ and $t \in [T, T^+]$ such that $p_1^t, \cdots, p_t^t$ are defined,\footnote{Recall in Case 4 of the Add Operation and Case 5 of the Exchange Operation, we go directly from $t$ to $t+2$ pivots, skipping phase $t+1$.} 
\begin{equation} \label{eq: well-sep-inner-induction}
p_1^t, \dots, p_t^t \text{ are } \beta_{t}\text{-well-separated w.r.t. } w_t. \tag{$\Diamond$}
\end{equation}
Moreover, at least one of the following properties holds:
\begin{enumerate}[label=(\alph*)]
\item $c_j^T$ is $\beta_{t+1}$-well-separated from $p_1^t, \dots, p_t^t$ w.r.t. $w_i$.
\item $c_j^T$ is $\beta_{t+1}$-attached to $p_j^t$ w.r.t. $w_t$. 
\item $c_j^T$ is $f(t,T)$-attached to $p_j^t$ w.r.t. $w_t$ and $w_t(c_j^T) < w_t(p_j^t)$, where $f(t,T) = {\beta_T \cdot (t-T)}$.
\end{enumerate}
\end{restatable}

For the proof sketch we focus on Case 4 of the Add Operation, which will give a flavor of the arguments. This is a concerning case a priori; for, if we were to add $x_{\alpha}$ to the set of pivots as in Cases 2 and 3, it is ambiguous as to whether $x_{\alpha}$ should be associated with label $f$ or $g$, as both $c_f^T$ and $c_g^T$ are close to $x_{\alpha}$ (see Appendix \ref{appendix: label_conflicts} for a diagram). We maneuver around the issue by making $c_f^T$ and $c_g^T$ new pivots and excluding $x_{\alpha}$. However, it is not immediately clear that such a step will preserve the desired invariants. To give intuition, we suppress the separation parameters and the precise weights used, though emphasize both are brittle (e.g., the arguments rely heavily on $\beta_t$ decreasing with $t$, see also Appendix \ref{appendix: decreasing_well_sep}). The directions of attachment between points (arrows in Figure \ref{fig: pro3}) are also crucial.  We will also see why we need to couple the induction with (a)---(c). 

To prove the inductive step for (\ref{eq: well-sep-inner-induction}) when Case 4 of the Add Operation is performed,  we need to show (i) $c_f^T$ and $c_g^T$ are well-separated, (ii), WLOG, $c_f^T$ is well-separated from $p_f^t$, and (iii), WLOG, $c_f$ is well-separated from $p_l^t$, $l \neq f$. See Figure \ref{fig: pro3}. When we say ``close'' or ``far'' below, we always mean in a weighted sense. For (i), because $p_f^T$ is close to $c_f^T$ (Proposition \ref{prop: attached_estimated_center}) and likewise for $p_g^T$, $c_g^T$, then $c_f^T$ and $c_g^T$ cannot be close, since this would violate that $p_f^T$ and $p_g^T$ are (inductively) far. To prove (ii), note $x_{\alpha}$ is far from $p_f^t$ by assumption of the Add Operation, and $c_f^T$ is close to $x_{\alpha}$ by assumption of Case 4, so $p_f^t$ and $c_f^T$ must be far. Finally for (iii), one can (inductively) deduce that (b) must hold when $j=f$, so $c_f^T$ and $p_f^t$ are close; but, since $p_f^t$ and $p_l^t$ are (inductively) far, $c_f^t$ and $p_l^t$ must be far. 

Proving the inductive step for (a)---(c) involves detailed casework. The Add and Exchange Operations are engineered so that, loosely speaking, an estimated center is either attached to the corresponding present pivot, or else breaks off to form its own pivot. A main subtlety is the direction and strength of attachment, e.g., property (c). Another is the sequence of operations, specifically, the Add Operation taking precedence over the Exchange Operation. 
\end{proof}

\begin{figure}[ht]
\centering
\includegraphics[width=0.7\textwidth]{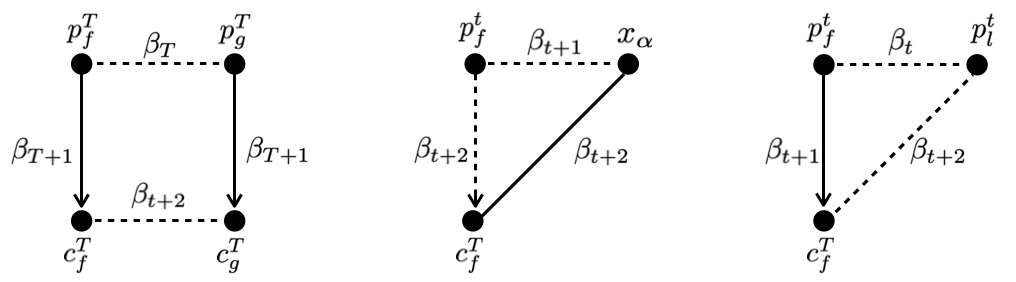}
\captionsetup{width=.9\linewidth}
\caption{Cases (i)---(iii) in the proof sketch of Lemma \ref{lem: well-sep-invariant}. Dashed lines indicate well-separation and solid lines indicate attachment, labelled with the appropriate parameters. Arrows go from smaller to larger natural weights.}
\label{fig: pro3}
\end{figure}

Theorem \ref{thm: under_k_labels} follows from Lemma \ref{lem: well-sep-invariant} and Proposition \ref{prop: less_than_k_centers} once we observe that we have set $\beta_1$ sufficiently large. For completeness, we include the proof below.

\begin{proof}[Proof of Theorem \ref{thm: under_k_labels}]
The number of labels used by the algorithm is the number of pivots in the last phase. By Lemma \ref{lem: well-sep-invariant}, we maintain the invariant that pivots $p_1^t, \dots, p_t^t$ are $\beta_t$-well-separated w.r.t. the natural weights at every time step in phase $t$. Suppose to the contrary that the final number of pivots is strictly more than $k$. Then at some point there are $t = k+1$ or $t = k+2$\footnote{The algorithm may skip a phase, hence we consider both cases.} pivots that are $\beta_t$-well-separated w.r.t. the natural weights throughout phase $t$. But $\beta_{k+2} = 8$, and it is impossible for $k+2$ points to be 8-well-separated, by Proposition \ref{prop: less_than_k_centers}. We conclude the final number of pivots is at most $k$, so the algorithm uses at most $k$ labels.
\end{proof}

\medskip

As the proof of Proposition \ref{prop: intermediate_attachment} shows,  the Add and Exchange operations are engineered so that the estimated center $c_j^{T^-}$ is close to $p_j^T$ at the beginning of phase $T$. Lemma \ref{lem: well_attached_centers} states that this property is maintained through the end of phase $T$. In essence, this is because no Add or Exchange operations are executed during phase $T$, so we can show that the attachment between $c_j^{T^-}$ and $p_j^T$ is static---even as natural weights increase. The proof is below. 

\medskip

\begin{proof}[Proof of Lemma \ref{lem: well_attached_centers}]
We need to show that: either $w^T(c_j^T) \leq w^{T^+}(p_j^{T^+})$ and $w^{T}(c_j^{T}) \cdot d(c_j^{T}, p_j^{T^+}) \leq \beta_T(T^+ - T) \cdot \textsf{OPT}$, \underline{or} $c_j^{T}$ is $\beta_{T^+ + 1}$-attached to $p_j^{T^+}$ w.r.t. $w^{T^+}$.

We need the second part of Proposition \ref{prop: intermediate_attachment}. When we reach the start of phase $T^+$, there are no Add Operations or Exchange Operations available, so (a) in the statement of Proposition \ref{prop: intermediate_attachment} cannot hold when $t = T^+$. Thus either (b) or (c) must hold.  

Let $w_{T^+}$ denote the natural weights at the \textit{start} of phase $T^+$. Note this notation is consistent with taking $t=T^+$ in $w_t$ in Proposition \ref{prop: intermediate_attachment}.  

\setcounter{case}{0}

\begin{case} \label{case: cs1-attachment-lemma}
In Proposition \ref{prop: intermediate_attachment}, (c) holds. 
\end{case}
If (c) holds, then 
\[ w^T(c_j^T) \leq w_{T^+}(c_j^T) < w_{T^+}(p_j^{T^+}) \leq w^{T^+}(p_j^{T^+}), \hspace{0.3cm} \text{and}\] 
\[w^T(c_j^T) \cdot d(c_j^T, p_j^{T^+}) \leq w_{T^+}(c_j^T) \cdot d(c_j^T, p_j^{T^+}) \leq \beta_T(T^+ - T) \cdot \textsf{OPT}\]
where the second inequality in the first line and the last inequality in the second line follow from (c) holding in Proposition \ref{prop: intermediate_attachment}. 

\begin{case} 
In Proposition \ref{prop: intermediate_attachment}, (c) does not hold. 
\end{case}
If (c) does not hold, then $w_{T^+}(c_j^T) \geq w_{T^+}(p_j^{T^+})$ and $c_j^T$ is $\beta_{{T^+}+1}$-attached to $p_j^{T^+}$ w.r.t. $w_{T^+}$, i.e., the weights at the \textit{beginning} of phase $T^+$ (since $\beta_{T^+ +1} < \beta_T(T^+ - T)$). We need to show that $c_j^T$ is $\beta_{{T^+}+1}$-attached to $p_j^{T^+}$ w.r.t. $w^{T^+}$, i.e., the weights at the \textit{end} of phase $T^+$. 
Recall that through the end of phase $T^+$, $c_j^T$ remains $\beta_{{T^+}+1}$-attached to at least one of the ${T^+}$ pivots (otherwise, the phase would terminate and an Add Operation would be executed). So it just remains to show that at the end of phase ${T^+}$, $c_j^T$ is still $\beta_{{T^+}+1}$-attached to $p_j^{T^+}$ in particular, w.r.t. $w^{T^+}$. Suppose to the contrary that $c_j^T$ is $\beta_{{T^+}+1}$-attached to $p_{j'}^{T^+}$ w.r.t. $w^{T+}$, where $j' \neq j$. Then it must be the case that $c_j^T$ is $\beta_{{T^+}+1}$-attached to $p_{j'}^{T^+}$ w.r.t. $w_{T+}$. But we also know that $w_{T^+}(c_j^T) \geq w_{T^+}(p_j^{T^+})$ and $c_j^T$ is $\beta_{{T^+}+1}$-attached to $p_j^{T^+}$ w.r.t. $w_{T^+}$. By Proposition \ref{prop: meta-prop}, this contradicts that $p_{j}^{T^+}$ and $p_{j'}^{T^+}$ are $\beta_{T^+}$-well-separated w.r.t. $w_{T^+}$.

\end{proof}

Having established Lemma \ref{lem: well_attached_centers}, we are ready to bound the cost of the algorithm, using the present pivot $p_j^T$ as the ``bridge'' between old and newly arriving points given label $j$. We will show inductively that, at bounded cost, we can move the old points to $c_j^{T^-}$, which is in some sense close to $p_j^T$ by Lemma \ref{lem: well_attached_centers}. In turn, $p_j^T$ dictates the greedy choices for the new points. So we will combine the cost of old and new points via $p_j^T$.

\section{Bounding the Algorithm's Cost} \label{sec:cost}

Throughout, let $cost(S; c) = \sum_{p \in S} d(p,c)$ for $S \subseteq X$ and $c \in X$.

As a first step, we need to bound the cost contribution of points that arrive during a single phase. The strategy is to compare the online greedy choices with the offline optimal choices, and to show these are sufficiently similar. More specifically, we know that in $D(T)$ each offline optimal center $y_i$ in $P_T$ is in the neighborhood of exactly one pivot, namely $p(y_i)$, and $y_i$ and $p(y_i)$ are close in a weighted sense, i.e., attached (Proposition \ref{prop: digraph_attachment}).  We further know that  since no Exchange Operations are executed during a phase, we can show that if $y_{i_1}$ and $y_{i_2}$ are in the neighborhood of the same pivot ($p(y_{i_1}) = p(y_{i_2})$), then $y_{i_1}$ and $y_{i_2}$ are also close in a weighted sense. 

Using these facts, we would be in good shape if we could show that for every point arriving during phase $T$, it is the case that if the point is assigned to $y_i$ in the offline optimal solution, then it receives the label of pivot $p(y_i)$ online. While this is not quite true, we can instead show that the number of points that do not satisfy this condition is small relative to the natural weights of their pivot, owing to the well-separated invariant (Lemma \ref{lem: well-sep-invariant}). Further, we show that these points can still be moved to their pivots at bounded cost, due to the greedy labelling rule. In effect, we will \textit{charge} the cost of these ``far'' points to their pivot. Lemma \ref{lem: far_points} summarizes this argument, and is used to prove the main theorem, Theorem \ref{thm: main_thm}.

\begin{lemma} \label{lem: far_points}
Let $T$ be a non-intermediate phase. For any $j \in [T]$, let $C_j$ be the points given label $j$ during phase $T$, i.e., $C_j = C_j^T \setminus C_j^{T^-}$. Define $S_{ji}$ to be be the set of elements in $C_j$ assigned to $y_i$ in the clustering of $X(T) \setminus X(T^-)$ induced by $P_T$. Define $S_{far, j} = \bigcup_{i : p(y_i) \neq p_j^T} S_{ji}$. Then 
\begin{enumerate}
\item $cost(S_{far,j}; p_j^T) \leq k \cdot (\beta_{T+1} + 2) \cdot \textsf{OPT}$, and 
\item $|S_{far,j}| \leq k \cdot w^T(p_j^T)$, where $w^T$ denotes the natural weights at the end of phase $T$.
\end{enumerate}
\end{lemma}

The following lemma states that the number of points in a cluster by the end of any phase is a bounded factor away from the natural weight (at the end of the phase) of the estimated center for that cluster at the end of the phase. 

\begin{lemma} \label{lem: sufficiently_weighted_centers}
Let $T$ be a non-intermediate phase and $j \in [T]$. Let $w^T(c_j^T)$ denote the natural weight of $c_j^T$ at the end of phase $T$ and let $C_j^T$ denote the set of points in cluster $j$ by the end of phase $T$. Then 
$$|C_j^T| \leq (2k+1)\cdot T \cdot w^T(c_j^T).$$

\end{lemma}

\begin{lemma}\label{lem: cross_phase_bounded_cost} 
Let $T$ be a non-intermediate phase and $j \in [T]$. Then  $cost(C_j^T)$ is bounded against center $c_j^T$, i.e., 
\[ \sum_{x \in C_j^T} d(x, c_j^T) \leq g(T,k) \cdot \textsf{OPT}, \hspace{0.2cm}
g(T,k) = T \cdot g(k),  \hspace{0.1cm} g(k) = \beta_1(2k^3 + 3k^2 +5k+1) + 2k + 4.\] 
\end{lemma}

As a corollary to Lemma \ref{lem: cross_phase_bounded_cost}, we have Theorem~\ref{thm: main_thm}, the main theorem.

Broadly, here is how the Lemmas \ref{lem: well_attached_centers}, \ref{lem: far_points}, and \ref{lem: sufficiently_weighted_centers} will tie together to prove Lemma \ref{lem: cross_phase_bounded_cost}. The proof of Lemma \ref{lem: cross_phase_bounded_cost} is by induction. Inductively, the points in cluster $j$ that arrived before phase $T$, which we call $C_j^{T^-}$, can be moved to their estimated center $c_j^{T^-}$ at bounded cost. This estimated center is close to $p_j^T$ by Lemma \ref{lem: well_attached_centers}. For instance, (a) in Lemma \ref{lem: well_attached_centers} says that once the points in $C_j^{T^-}$ are moved to $c_j^{T^-}$, they can also be moved to $p_j^T$ at bounded cost, since $w^{T^-}(c_j^{T^-}) \cdot d(c_j^{T^-}, p_j^T)$ is bounded and $|C_j^{T^-}| \leq h(T^-, k) \cdot w^{T^-}(c_j^{T^-})$ (Lemma \ref{lem: sufficiently_weighted_centers}). Finally, since $w^{T^-}(c_j^{T^-}) \leq w^T(p_j^T)$, we will be able to \textit{charge} the points in $C_j^{T^-}$ to $p_j^T$ in order to move them to $c_j^T$, which is attached to $p_j^T$ (Proposition \ref{prop: attached_estimated_center}). By similar logic, the cost of the far points $S_{far, j}$ from phase $T$ can be charged to $p_j^T$ and then moved to $c_j^T$ (Lemma \ref{lem: far_points}). Finally, the remaining points given label $j$ in phase $T$, call them $S_{near,j}$, are close to offline centers in $P_T$ that are in turn close to $p_j^T$. Crucially, no Exchange Operation is executed during a phase, so these offline centers are also close to $c_j^T$ in a weighted sense.

\subsection{Proofs for bounding cost} \label{sec: proofs_bounded_cost}

We start with the proof of Lemmas \ref{lem: far_points}, which is a key step where the greedy rule for assigning labels to points is used. 

\begin{proof}[Proof of Lemma \ref{lem: far_points}]
 WLOG, let $j = T$. For $c \in P_T$, let $m(c)$ be the number of points assigned to $c$ in the clustering of $X(T) \setminus X(T^-)$ induced by the centers $P_T$, i.e., in this clustering every point is assigned to the \textit{nearest} point in $P_T$. 
 
 For shorthand, let $w$ denote the natural weights $w^T$ of points at the end of phase $T$.

\begin{observation} \label{obs1}
For $c \in P_T$, $w(c) \geq m(c)$. 
\end{observation}

This follows from the definition of $w(c)$ and the fact that there are $m(c)$ points whose movement cost to $c$ is at most $2\textsf{OPT}$, by construction of $P_T$.

\begin{observation} \label{obs2}
If $(p(y_i), y_i)$ is a directed edge in $D(T)$, then $w(p(y_i)) \cdot d(p(y_i), y_i) < \beta_{T+1} \cdot \textsf{OPT}$. Likewise, if $(y_i, p(y_i))$ is a directed edge in $D(T)$, then $w(y_i) \cdot d(p(y_i), y_i) < \beta_{T+1} \cdot  \textsf{OPT}$. 
\end{observation}
This follows from the definition of $D(T)$ and Proposition \ref{prop: digraph_attachment}. 

 Call the points in $S_{far,T}$ \textit{far} points. In the claims below, we show that the far points can be moved to $p_T^T$ at bounded cost (Claims \ref{clm: phase-wise_bdd_clm1} and \ref{clm: phase-wise_bdd_clm2}), and that there are not too many far points relative to the weight of $p_T^T$ (Claim \ref{clm: phase-wise_bdd_clm3}). In turn, we will be able to \textit{charge} the cost of the far points to $p_T^T$.   

\setcounter{claim}{0}
\begin{claim} \label{clm: phase-wise_bdd_clm1}
Let $p(y_i) \neq p_T^T$. Suppose $w(y_i) > w(p(y_i))$. Then $cost(S_{Ti};p_T^T) \leq {(\beta_{T+1} + 2)}\textsf{OPT}$.  
\end{claim}

\begin{proof}
WLOG, let $p(y_i) = p_1^T$. We consider two cases. 

\setcounter{case}{0}

\begin{case} \label{cs1}
$|S_{Ti}| \geq w(p_1^T)$. We will show this case cannot happen.
\end{case}
In this case, we know that $w(y_i) \geq m(y_i) \geq |S_{Ti}|\geq w(p_1^T)$. We know by Observation \ref{obs2} that ${w(p_1^T) \cdot d(p_1^T, y_i) < \beta_{T+1}\cdot \textsf{OPT}}$. By Proposition \ref{prop: meta-prop}, this implies $w(p_1^T) \cdot d(y_i, p_T^T) \geq 2\beta_{T+1} \cdot \textsf{OPT}$. 

Since $|S_{Ti}| \geq w(p_1^T)$, there exists $S_{Ti}' \subseteq S_{Ti}$ such that $|S_{Ti}'| = w(p_1^T)$. In turn, $cost(S_{Ti}'; p_1^T) \leq cost(S_{Ti}'; y_i) + w(p_1^T) \cdot d(y_i,p_1^T) <  (\beta_{T+1} + 2)\cdot \textsf{OPT}$, since $P_T$ is a clustering with cost at most $2\textsf{OPT}$. On the other hand, 
$$cost(S_{Ti}'; p_T^T) \geq \sum_{p \in S_{Ti}'} d(y_i, p_T^T) - \sum_{p \in S_{Ti}'} d(p, y_i) = w(p_1^T) \cdot d(y_i, p_T^T) - \sum_{p \in S_{Ti}'} d(p,y_i) \geq (2\beta_{T+1} - 2)\textsf{OPT}.$$ 

Since $\beta_{T+1} \geq 4$, $\beta_{T+1} +2 \leq 2\beta_{T+1} - 2$, so $cost(S_{Ti}';p_1^T) < cost(S_{Ti}'; p_T^T)$, which violates that $T = \arg \min_{j \in [T]} d(p, p_j^T)$ for all $p \in S_{Ti}' \subseteq C_T$.

\begin{case}
$|S_{Ti}| \leq w_t(p_1^T)$. 
\end{case}
In this case, we know that since $w(p_1^T) \cdot d(y_i, p_1^T) < \beta_{T+1} \cdot \textsf{OPT}$, we also have $|S_{Ti}| \cdot d(y_i, p_1^T) < \beta_{T+1} \cdot \textsf{OPT}$. By the triangle inequality, 
\[cost(S_{Ti};p_1^T) \leq cost(S_{Ti};y_i) + |S_{Ti}|\cdot d(y_i, p_1^T) \leq 2\textsf{OPT} + \beta_{T+1} \cdot \textsf{OPT}. \]
Since $cost(S_{Ti};p_T^T) \leq cost(S_{Ti};p_1^T)$ by the greedy procedure, this proves Claim \ref{clm: phase-wise_bdd_clm1}. 
\end{proof}

\begin{claim} \label{clm: phase-wise_bdd_clm2}
Let $p(y_i) \neq p_T^T$. Suppose that $w(y_i) \leq w(p(y_i))$. Then $cost(S_{Ti};p_T^T) \leq {(\beta_{T+1} +1)}\textsf{OPT}$.
\end{claim}

\begin{proof}
WLOG, let $p(y_i) = p_1^T$. By Observation \ref{obs2}, $w(y_i)\cdot d(y_i, p_1^T) < \beta_{T+1} \cdot \textsf{OPT}$. Further, $|S_{Ti}| \leq m(y_i) \leq w(y_i)$, so $|S_{Ti}|\cdot d(y_i,p_1^T) < \beta_{T+1} \cdot \textsf{OPT}$. So:
$$cost(S_{Ti};p_T^T) \leq cost(S_{Ti};p_1^T) \leq cost(S_{Ti}; y_i) + |S_{Ti}| \cdot d(y_i, p_1^T) \leq 2 \textsf{OPT} + \beta_{T+1} \cdot \textsf{OPT}.$$
\end{proof}

\begin{claim} \label{clm: phase-wise_bdd_clm3}
 Let $p(y_i) \neq p_T^T$. Then $|S_{Ti}| \leq w(p_T^T)$.
\end{claim}

\begin{proof}

\setcounter{case}{0}

As before, assume WLOG that $p(y_i) = p_1^T$.
\begin{case} \label{cs3}
$w(y_i) > w(p_1^T)$. 
\end{case}
We know from the proof of Claim \ref{clm: phase-wise_bdd_clm1}, Case \ref{cs1} that this implies $|S_{Ti}| < w(p_1^T)$. We have 
\begin{align*}
    |S_{Ti}| \cdot d(p_T^T, y_i) &= \sum_{p \in S_{Ti}}d(y_i, p_T^T) \\
    &\leq \sum_{p \in S_{Ti}} d(p,p_T^T) + \sum_{p \in S_{Ti}} d(p, y_i) \\
    &\leq (\beta_{T+1} + 2)\textsf{OPT} + 2\textsf{OPT} \tag{\text{Claim \ref{clm: phase-wise_bdd_clm1}}}\\
    &\leq 2\beta_{T+1} \cdot \textsf{OPT} \\ 
    &\leq w(p_T^T) \cdot d(p_T^T, y_i) 
\end{align*}
where in the last line we have applied Proposition \ref{prop: meta-prop}, using that $w(y_i) > w(p_1^T)$, Observation \ref{obs2}, and $p_1^T$ and $p_T^T$ are $\beta_T$-well-separated w.r.t. $w$. Finally, dividing both ends of the chain of inequalities by $d(p_T^T, y_j)$ gives $|S_{Tj}| \leq w(p_T^T)$, as desired. 

\begin{case}
$w(y_i) \leq w(p_1^T)$.
\end{case}
First consider when $w(p_T^T) \geq w(y_i)$. Then $w(p_T^T) \geq w(y_i) \geq m(y_i) \geq |S_{Ti}|$, so the claim follows. 

So the last case to consider is when $w(p_T^T) < w(y_i)$. It suffices to show that $w(p_T^T) \cdot d(p_T^T, y_i) \geq 2\beta_{T+1} \cdot \textsf{OPT}$; then, we can just apply the argument in Case \ref{cs3}. Suppose to the contrary that $w(p_T^T) \cdot d(p_T^T, y_i) < 2\beta_{T+1} \cdot \textsf{OPT}$. Then 
\begin{align*}
    \beta_T \cdot \textsf{OPT} &\leq  w(p_T^T) \cdot d(p_T^T, p_1^T) \\
    &\leq w(p_T^T) \cdot d(p_T^T, y_i) + w(p_T^T) \cdot d(y_i, p_1^T) \\
    &\leq 2\beta_{T+1} \cdot \textsf{OPT} +  w(p_T^T) \cdot d(y_i, p_1^T) \\
    &< 2\beta_{t+1} \cdot \textsf{OPT} + w(y_i) \cdot d(y_i, p_1^T) \\
    &< 2\beta_{T+1} \cdot \textsf{OPT} + \beta_{T+1} \cdot \textsf{OPT} \\
    &= \beta_T \cdot \textsf{OPT}
\end{align*}
where the second-to-last line follows from Observation \ref{obs2}. The left-hand and right-hand sides give a contradiction, concluding the proof of the case and the claim. 
\end{proof}

\begin{claim} \label{clm4}
$cost(S_{far,T};p_T^T) \leq k \cdot (\beta_{T+1} + 2)\textsf{OPT}$ and $|S_{far,T}| \leq k \cdot w(p_T^T)$. 
\end{claim}

\begin{proof}
By Claims \ref{clm: phase-wise_bdd_clm1} and \ref{clm: phase-wise_bdd_clm2}, 
$$cost(S_{far,T};p_T^T) = \sum_{i:p(y_i) \neq p_T^T} cost(S_{Ti};p_T^T) \leq k \cdot (\beta_{T+1} + 2)\textsf{OPT}$$
By Claim \ref{clm: phase-wise_bdd_clm3}, 
$$|S_{far,T}| = \sum_{i: p(y_i) \neq p_T^T} |S_{Ti}| \leq k\cdot w(p_T^T).$$
\end{proof}
\end{proof}

The proof of Lemma \ref{lem: sufficiently_weighted_centers} can be found in Appendix \ref{appendix: proofs_bounded_cost}. The argument is by induction and is similar in flavor to the proof of Lemma \ref{lem: cross_phase_bounded_cost}, which we give below. Lemma \ref{lem: cross_phase_bounded_cost} implies Theorem \ref{thm: main_thm}. 
\begin{proof}[Proof of Lemma \ref{lem: cross_phase_bounded_cost}]
 The proof is by induction. Let $C_j$, $S_{ji}$, and $S_{far, j}$ be as in Lemma \ref{lem: far_points}. Define $S_{near, j} = \bigcup_{i : p(y_i) = p_j^T} S_{ji}$ and $S_j$ to be the elements in $C_j$ that are assigned to $p_j^T$ in the clustering of $X(T) \setminus X(T^-)$ induced by $P_T$. Let $w^t$ denote the natural weights at the end of phase $t$. First we need the following key claim. 

 \setcounter{claim}{0}
\begin{claim} \label{clm: main_lem_clm1}
For any $x,y \in \delta^+(p_j^T) \cup \delta^-(p_j^T) \cup \{p_j^T\}$, $x$ and $y$ are $2\beta_{T+1}$-attached w.r.t. $w^T$. 
\end{claim}

\begin{proof}[Proof of Claim \ref{clm: main_lem_clm1}]
If $x$ or $y$ is $p_j^T$, then the claim automatically holds by Proposition \ref{prop: digraph_attachment}. There are two other cases. The first case is, WLOG, $x \in \delta^-(p_j^T)$. Regardless of whether $y$ is in $\delta^{-}(p_j^T)$ or $\delta^{+}(p_j^T)$, the claim holds by Propositions \ref{prop: digraph_attachment} and \ref{prop: meta-prop}. The second case is that $x,y \in \delta^{+}(p_j^T)$. We prove the stronger statement that $x$ and $y$ are $\beta_{T+1}$-attached w.r.t. $w^T$. Suppose to the contrary that $x$ and $y$ are $\beta_{T+1}$-well-separated. We claim that this implies 
\begin{equation}\label{add_swap_contr}
\{p_1^T, \dots, p_T^T\} \cup \{x, y \} \setminus \{p_j^T\}
\end{equation}
is $\beta_{T+1}$-well-separated w.r.t. $w^T$; this would give a contradiction, since if an Exchange Operation were available, it would have been executed. Now suppose that (\ref{add_swap_contr}) does not hold. Then WLOG $p_{j'}^T$ and $x$ are $\beta_{T+1}$-attached w.r.t. $w^T$, for some $j' \neq j$. Since $x \in \delta^+(p_j^T)$ and since $x$ and $p_j^T$ are $\beta_{T+1}$-attached w.r.t. $w^T$, by Proposition \ref{prop: meta-prop}, $p_j^T$ and $p_{j'}^T$ are $2\beta_{T+1}$-attached w.r.t. $w^T$. This contradicts that $p_j^T$ and $p_{j'}^T$ are $\beta_T$-well-separated w.r.t. $w^T$, since $2\beta_{T+1} < \beta_T$. This concludes the proof of the case and the claim.
\end{proof}
To bound the cost contribution of $C_j^{T^-}$, we case on which statement holds in Lemma \ref{lem: well_attached_centers}.

\setcounter{case}{0}
\begin{case} \label{cs: main_lem_cs1}
$c_j^{T^-}$ is $\beta_{T+1}$-attached to $p_j^T$ w.r.t. $w^T$ (i.e., (b) holds in Lemma \ref{lem: well_attached_centers}). 
\end{case}

Since in Case \ref{cs: main_lem_cs1}, $c_j^{T^-}$ is $\beta_{T+1}$-attached to $p_j^T$ w.r.t. $w^T$, $c_j^{T^-} \in \delta^+(p_j^T) \cup \delta^-(p_j^T)$. Also, $c_j^T$ by definition is in $\delta^+(p_j^T) \cup \{p_j^T\}$. So by Claim \ref{clm: main_lem_clm1}, $c_j^{T^-}$ is $2\beta_{T+1}$-attached to $c_j^T$ w.r.t. $w^T$. Using this, we bound $cost(C_j^{T^-}; c_j^T)$:
\begin{align}
cost(C_j^{T^-}; c_j^T) &\leq cost(C_j^{T^-}; c_j^{T^-}) + |C_j^{T^-}| \cdot d(c_j^{T^-}, c_j^T) \notag \\
&\leq g(T^-, k) \cdot \textsf{OPT} + |C_j^{T^-}| \cdot d(c_j^{T^-}, c_j^T) \notag \\
&\leq  g(T^-, k) \cdot \textsf{OPT} + (2k+1) \cdot T^- \cdot w^{T^-}(c_j^{T^-}) \cdot d(c_j^{T^-}, c_j^T)  \notag \\
&\leq  g(T^-, k) \cdot \textsf{OPT} + (2k+1) \cdot T^- \cdot w^{T}(c_j^{T^-}) \cdot d(c_j^{T^-}, c_j^T)  \notag \\
&\leq g(T^-, k) \cdot \textsf{OPT} + (2k+1) \cdot T^- \cdot 2 \beta_{T+1} \cdot \textsf{OPT} \label{prev_cost_cs1}
\end{align}
where the third inequality is due to Lemma \ref{lem: sufficiently_weighted_centers}.

\begin{case} \label{cs: main_lem_cs2}
(b) does not hold in Lemma \ref{lem: well_attached_centers}, so (a) holds, i.e., $w^{T^-}(c_j^{T^-}) \leq w^T(p_j^T)$ and $w^{T^-}(c_j^{T^-}) \cdot d(c_j^{T^-}, p_j^T) \leq \beta_{T^-}(T-T^-) \cdot \textsf{OPT}$. 
\end{case}

We bound $cost(C_j^{T^-}; c_j^T)$:
\begin{align}
cost(C_j^{T^-}; c_j^T) &\leq cost(C_j^{T^-}; c_j^{T^-}) + |C_j^{T^-}| \cdot d(c_j^{T^-}, c_j^T) \notag \\
&\leq g(T^-, k) \cdot \textsf{OPT} + |C_j^{T^-}| \cdot d(c_j^{T^-}, p_j^T) + |C_j^{T^-}| \cdot d(p_j^T, c_j^T) \label{two_bounds}
\end{align}
and now we use the assumptions of the case to continue bounding from (\ref{two_bounds}):
\begin{align}
|C_j^{T^-}| \cdot d(c_j^{T^-}, p_j^T) &\leq (2k+1) \cdot T^- \cdot w^{T^-}(c_j^{T^-}) \cdot d(c_j^{T^-}, p_j^T) \notag \\
&\leq (2k+1) \cdot T^- \cdot \beta_{T^-} (T-T^{-}) \cdot \textsf{OPT} \label{bound_1}
\end{align}
where the first inequality is due to Lemma \ref{lem: sufficiently_weighted_centers}. Next, 
\begin{align}
|C_j^{T^-}| \cdot d(p_j^T, c_j^T) \leq (2k+1) T^- \cdot w^{T^-}(c_j^{T^-}) \cdot d(p_j^T, c_j^T) &\leq (2k+1)  T^- \cdot w^{T}(p_j^{T}) \cdot d(p_j^T, c_j^T) \notag \\
&\leq (2k+1) T^- \cdot \beta_{T+1} \cdot \textsf{OPT} \label{bound_2}
\end{align}
where the first inequality is due to Lemma \ref{lem: sufficiently_weighted_centers} and the last inequality is due to Proposition \ref{prop: attached_estimated_center}. So combining (\ref{two_bounds}), (\ref{bound_1}), (\ref{bound_2}) gives 
\begin{equation} \label{prev_cost_cs2}
    cost(C_j^{T^-}; c_j^T) \leq g(T^-, k) \cdot \textsf{OPT} + (2k+1) \cdot T^- \cdot (\beta_{T^-}(T-T^-) + \beta_{T+1}) \cdot \textsf{OPT}. 
\end{equation}
Now we have bounds (\ref{prev_cost_cs1}) and (\ref{prev_cost_cs2}) for $cost(C_j^{T^-}; c_j^T)$. Recall that  $C_j^T = C_j^{T^-} \cup S_{far, j} \cup S_{near, j} \cup S_j$. The following bounds will hold regardless of whether we are in Case \ref{cs: main_lem_cs1} or \ref{cs: main_lem_cs2}. We have
\begin{equation} \label{pivot_cost}
cost(S_j; c_j^T) \leq cost(S_j; p_j^T) + |S_j| \cdot d(p_j^T, c_j^T) \leq 2\textsf{OPT} + w^T(p_j^T) \cdot d(p_j^T, c_j^T) \leq (2 + \beta_{T+1})\textsf{OPT}
\end{equation}
\begin{align}
cost(S_{near,j}; c_j^T) &= \sum_{i: p(y_i) = p_j^T} cost(S_{ji}; c_j^T)
\leq \sum_{i: p(y_i) = p_j^T} \sum_{p \in S_{ji}} d(p, c_j^T) \notag \\
&\leq 2\textsf{OPT} + \sum_{i: p(y_i) = p_j^T} w^T(y_i)\cdot d(y_i, c_j^T)  \leq (2k\beta_{T+1} +2)\textsf{OPT} \label{near_cost}
\end{align}
where we have used Claim \ref{clm: main_lem_clm1} and that $|S_{ji}| \leq w^T(y_i)$. Finally, by Lemma \ref{lem: far_points},
\begin{align}
    cost(S_{far,j}; c_j^T) \leq cost(S_{far,j}; p_j^T) + |S_{far,j}| \cdot d(p_j^T, c_j^T) \leq k(2\beta_{T+1} +2) \textsf{OPT} \label{far_cost}
\end{align}
 Combining (\ref{pivot_cost}), (\ref{near_cost}), (\ref{far_cost}) with (\ref{prev_cost_cs1}) or (\ref{prev_cost_cs2}) gives the sought bound:
$$cost(C_j^T; c_j^T) \leq [g(T^-, k) + g(k)] \textsf{OPT} \leq g(T,k) \cdot \textsf{OPT}.$$
\end{proof}

\printbibliography

@inproceedings{FLN21,
  author    = {Hendrik Fichtenberger and Silvio Lattanzi and Ashkan Norouzi-Fard and Ola Svensson},
  title     = {Consistent k-Clustering for General Metrics},
  booktitle = {{ACM-SIAM} Symposium on Discrete
               Algorithms},
  year      = {2021},
}

@inproceedings{chanKDD,
author = {Chan, T-H. Hubert and Guerqin, Arnaud and Sozio, Mauro},
title = {Fully Dynamic K-Center Clustering},
year = {2018},
booktitle = { World Wide Web Conference},
pages = {579–587},
}

@article{AGKMMP01,
  author    = {Vijay Arya and
               Naveen Garg and
               Rohit Khandekar and
               Adam Meyerson and
               Kamesh Munagala and
               Vinayaka Pandit},
  title     = {Local Search Heuristics for k-Median and Facility Location Problems},
  journal   = {{SIAM} Journal of Computing},
  volume    = {33},
  number    = {3},
  pages     = {544--562},
  year      = {2004},
}

@article{CGST99,
  author    = {Moses Charikar and
               Sudipto Guha and
               {\'{E}}va Tardos and
               David B. Shmoys},
  title     = {A Constant-Factor Approximation Algorithm for the k-Median Problem},
  journal   = {Journal of Computer and Systems Sciences},
  volume    = {65},
  number    = {1},
  pages     = {129--149},
  year      = {2002},
}

@article{JV99,
	author = {K. Jain and V. V. Vazirani},
	title = {Approximation Algorithms for Metric Facility Location
	and k-Median Problems Using the Primal-Dual Schema
	and Lagrangian Relaxation},
	journal = {Journal of the ACM},
	volume = {48},
	number = {2},
	pages = {274 -- 296},
	year = {2001}
	}

@article{BPRST17,
  author    = {Jaroslaw Byrka and
               Thomas W. Pensyl and
               Bartosz Rybicki and
               Aravind Srinivasan and
               Khoa Trinh},
  title     = {An Improved Approximation for k-Median and Positive Correlation
               in Budgeted Optimization},
  journal   = {{ACM} Transactions on Algorithms},
  volume    = {13},
  number    = {2},
  pages     = {23:1--23:31},
  year      = {2017},
}

@article{LS13,
  author    = {Shi Li and
               Ola Svensson},
  title     = {Approximating k-Median via Pseudo-Approximation},
  journal   = {{SIAM} Journal of Computing},
  volume    = {45},
  number    = {2},
  pages     = {530--547},
  year      = {2016},
}

@inproceedings{LattanziV17,
  author    = {Silvio Lattanzi and
               Sergei Vassilvitskii},
  editor    = {Doina Precup and
               Yee Whye Teh},
  title     = {Consistent k-Clustering},
  booktitle = {International Conference on Machine Learning},
  pages     = {1975--1984},
  year      = {2017}
}

@inproceedings{Cohen-AddadHPSS19,
  author    = {Vincent Cohen{-}Addad and
               Niklas Hjuler and
               Nikos Parotsidis and
               David Saulpic and
               Chris Schwiegelshohn},
  title     = {Fully Dynamic Consistent Facility Location},
  booktitle = {Conference
               on Neural Information Processing Systems },
  pages     = {3250--3260},
  year      = {2019}
}

@inproceedings{Meyerson2001,
 author = {Meyerson, A.},
 title = {Online Facility Location},
 booktitle = { IEEE Symposium on Foundations of Computer Science},
 year = {2001},
  pages     = {426--431},
}

@misc{LSSTalk,
  title={k-means clustering},
  author={Liberty, Edo and Sriharsha, Ram and Sviridenko, Maxim},
  note = {talk slides},
  howpublished = "\url{https://www.cs.yale.edu/homes/el327/papers/OnlineKMeansAlenexEdoLiberty.pdf}",
}

@inproceedings{liberty2016algorithm,
  title={An algorithm for online k-means clustering},
  author={Liberty, Edo and Sriharsha, Ram and Sviridenko, Maxim},
  booktitle={ Workshop on Algorithm Engineering and Experiments },
  pages={81--89},
  year={2016},
}

@book{bishop,
    author = {Bishop, Christopher M.},
    title = {Pattern Recognition and Machine Learning},
    year = {2006},
    publisher = {Springer-Verlag},
    address = {Berlin, Heidelberg}
    }

@InProceedings{pmlr-v130-guo21a,
  title = 	 { Consistent k-Median: Simpler, Better and Robust },
  author =       {Guo, Xiangyu and Kulkarni, Janardhan and Li, Shi and Xian, Jiayi},
  booktitle = 	 { International Conference on Artificial Intelligence and Statistics},
  pages = 	 {1135--1143},
  year = 	 {2021},
  volume = 	 {130},
}

@article{CharikarCFM04,
  author    = {Moses Charikar and
               Chandra Chekuri and
               Tom{\'{a}}s Feder and
               Rajeev Motwani},
  title     = {Incremental Clustering and Dynamic Information Retrieval},
  journal   = {{SIAM} Journal of Computing},
  volume    = {33},
  number    = {6},
  pages     = {1417--1440},
  year      = {2004},
}

@book{roughgarden_2021, 
author = {Tim Roughgarden},
place={Cambridge}, 
title={Beyond the Worst-Case Analysis of Algorithms}, 
publisher={Cambridge University Press}, 
year={2021}}

@article{meyerson2004k,
  title={A k-median algorithm with running time independent of data size},
  author={Meyerson, Adam and O'Callaghan, Liadan and Plotkin, Serge},
  journal={Machine Learning},
  volume={56},
  number={1},
  pages={61--87},
  year={2004},
  publisher={Springer}
}

\appendix

\section{Helper Propositions} \label{appendix: helper_props}

\begin{fact}[Fact 2.1 in \cite{meyerson2004k}] \label{fact: Meyerson}
Let $N$ be a set of points, with $S \subseteq N$. Let $k$ be an integer with $0 \leq k \leq n$. Let $K \subseteq S$ be the $k$-element subset of $S$ minimizing $\sum_{x \in S} d(x,K)$ where $d(\cdot, \cdot)$ is the distance function on $N$, and $d(x,K)$ denotes $\min_{m \in K} d(x,m)$. Then if $K'$ is a $k$-element subset of $N$, $\sum_{x \in S} d(x,K) \leq 2 \sum_{x \in S} d(x,K')$. 
\end{fact}

\begin{proposition} \label{prop: meta-prop}
Let $x,y,p$ be three points, and let $w$ denote some associated weights. Assume that $\beta, \beta_x, \beta_y, \textsf{OPT} > 0$. Suppose that $w(x) \leq w(p)$ and that $x$ and $y$ are $(\beta, \textsf{OPT})$-well-separated w.r.t. $w$. If $x$ and $p$ are $(\beta_x, \textsf{OPT})$-attached w.r.t. $w$, and $y$ and $p$ are $(\beta_y, \textsf{OPT})$-attached w.r.t. $w$, then $\beta < \beta_x + \beta_y$. 
\end{proposition}

\begin{proof}[Proof of Proposition \ref{prop: meta-prop}]
We consider two cases. First, suppose that $w(p) \leq w(y)$. Then 
\begin{align*}
\beta \cdot \textsf{OPT} & \leq w(x) \cdot d(x,y) \\
& \leq w(x) \cdot d(x,p) + w(x) \cdot d(p,y) \\
&\leq w(x) \cdot d(x,p) + w(p) \cdot d(p,y) \\
&< \beta_x \cdot \textsf{OPT} + \beta_y \cdot \textsf{OPT}.
\end{align*}
For the second case, $w(y) \leq w(p)$. Then 
\begin{align*}
\beta \cdot \textsf{OPT} &\leq \min\{w(x), w(y)\} \cdot d(x,y) \\
&\leq \min\{w(x), w(y)\} \cdot d(x,p) + \min\{w(x), w(y)\} \cdot d(y,p) \\
&\leq w(x) \cdot d(x,p) + w(y) \cdot d(y,p) \\
&< \beta_x \cdot \textsf{OPT} + \beta_y \cdot \textsf{OPT}.
\end{align*}
\end{proof}

\section{Lower Bound} \label{appendix: lower_bd}

\begin{proof}[Proof of Theorem \ref{thm: lower_bd}]
For sake of contradiction, assume there is an online algorithm  $A$ that is a $c$-approximation for $c\leq \frac{k-1}{2} -\epsilon$ and any $\epsilon >0$. The algorithm is given a budget $\textsf{OPT}$.  The sequence of points that arrives will lie in $\mathbb{R}^k$. First there are $10k^2$ points $P_0$ that arrive at location $ (0,0,\dots, 0)$.  We may assume that $A$ assigns each of these the same label.  

There are $k-1$ remaining phases for how points are released. Index these phases as $1,2, \ldots k-1$. At the end of the $i$th phase, $A$ will have accumulated cost at least $\frac{i}{2} \textsf{OPT}$.  In the end, $A$'s cost will be at least $\frac{k-1}{2}\textsf{OPT}$, contradicting the assumption that $A$ is a $c$-approximation. Additionally, by the end of the $i$th phase, the algorithm will have used at least $i+1$ labels.  All points in the $i$th phase will arrive at location $(0, 0, \ldots \textsf{OPT}, \ldots 0)$ where the $i$th dimension is non-zero.  The number of points that arrive will be $10 k^2$ at this location.  The majority of these points will have to be labelled with a single label by the algorithm. 

Consider phase $i$.  The adversary releases one point $q_i$ at location $(0, 0, \ldots \textsf{OPT}, \ldots 0)$.  First \emph{assume} the algorithm the algorithm does not give this  point a new label. We will revisit this assumption at the end of the proof. Then regardless of what label the algorithm gives this point, it is distance at least $\textsf{OPT}$ from all other points that have arrived.  Thus, this point will add an additional $\textsf{OPT}/2$ cost to the algorithm.  Inductively, the total algorithm cost is at least $\frac{i}{2} \textsf{OPT}$.  

Next $10k^2 -1$ points arrive at this same location.  If the algorithm does not label half of these points with a new label, then these points will be grouped with points that arrived previously.  Since prior points are distances at least $\textsf{OPT}$ away, this will make the algorithm's total cost larger than $\Omega(k^2 \textsf{OPT})$. Thus, the algorithm must label half of these points with a new label.  Now the inductive invariants are satisfied.

At the end, the optimal solution labels the points that arrive in each phases a different label. The total optimal cost is $0$ and the algorithm has cost $\Omega(k \textsf{OPT})$.

Now we revisit the assumption that $q_i$ must not be given a new label by $A$.  Indeed, say it is given a new label.  Then the above procedure terminates at this time.  Instead, we know the algorithm has used $i+1$ labels.  There is a clustering of the points that have arrived whose cost is below \textsf{OPT} and only uses $i$ labels.  Intuitively, this is a large mistake. 

Indeed, consider the following. There will be $k-i$ points $P^*$ that will arrive.  The $\ell$th point arrives at location $(0, 0, \ldots L \cdot \textsf{OPT}, \ldots 0)$ where $L$ is a parameter. Here $\ell \in [k-i]$. The non-zero dimension is the $(i+\ell)th$ dimension. 

The optimal solution labels $q_i$ the same label as the points in $P_0$ and this cluster has cost \textsf{OPT}.  All other points get a unique label depending on which dimension is non-zero.  $k-1$ dimensions are used, combined with $P_0$ this is $k$ labels. The cost of these points are $0$ and the only cost is $q_i$ paying cost $\textsf{OPT}$.  The optimal cost is bounded. 

On the other hand, $A$  only has $k-i-1$ unused labels when points in $P^*$ begin arriving.  One of them must be given a label as some other point whose distance is at least $L\textsf{OPT}$ away.  The cost of this cluster is at least $\frac{L}{2} \textsf{OPT}$.  Setting $L$ large, contradicts the bounded approximation ratio of $A$.

\end{proof}

\section{Omitted Proofs from Section \ref{sec:invariant}} \label{appendix: omitted_proofs_invariant}

\begin{proof}[Proof of Proposition \ref{prop: attached_estimated_center}]
By the definition (\ref{eq: good_center}) of $c_j^T$, we know that $c_j^T \in \delta^+(p_j^T) \cup p_j^T \subseteq \\ \{y_1, \dots, y_l, p_j^T, c_j^{T^-}\}$. By Proposition \ref{prop: digraph_attachment} and the fact that $c_j^{T^-} \in \delta^+(p_j^T)$ only if $c_j^{T^-}$ is $\beta_{T+1}$-attached to $p_j^T$ w.r.t. $w^T$, $c_j^T$ is $\beta_{T+1}$-attached to $p_j^T$ w.r.t. $w^T$. That $w^T(c_j^T) \geq w^T(p_j^T)$ follows directly from $c_j^T$ maximizing $w^T$ over a set that includes $p_j^T$, and this is with equality if and only if $c_j^T = p_j^T$ due to the tiebraking rule in the definition (\ref{eq: good_center}) of estimated center. 
\end{proof}

\begin{proof}[Proof of Lemma \ref{lem: well-sep-invariant}]
The proof will consist of an inner induction nested in an outer induction. Let $1 = T_1 < T_2 < \cdots < T_L \leq k$ denote the non-intermediate phases (i.e., phases during which at least one point receives a label). 
Take $T = T_i$, so that $T^+ = T_{i+1}$. In the \textit{outer induction}, we will prove Lemma \ref{lem: well-sep-invariant} for the non-intermediate phases: i.e., if $p_1^T, \dots, p_T^T$ are $\beta_{T}$-well-separated w.r.t. the natural weights at the start of phase $T$, then $p_1^{T^+}, \dots, p_{T^+}^{T^+}$ are $\beta_{T^+}$-well-separated w.r.t. the natural weights at the start of phase $T^+$. Note the base case is satisfied, since $\{p_1^1\}$ is trivially a $\beta_1$-well-separated set w.r.t. the natural weights at the start of phase 1, thus also w.r.t. the natural weights at the end of phase 1. Given the inductive assumption that $p_1^T, \cdots, p_T^T$ are $\beta_T$-well-separated w.r.t. the natural weights at the start of phase $T$, we prove the following claim. Taking $t=T^+$ in the claim completes the proof.

\innerinduct*

 Note that if (a) holds, another Add Operation is available, so $t < T^+$.
 
 The proof of Proposition \ref{prop: intermediate_attachment} will be the \textit{inner induction}. The first part (\ref{eq: well-sep-inner-induction}) of the claim establishes that Lemma \ref{lem: well-sep-invariant} holds for the intermediate phases. Note that while we only need to prove (\ref{eq: well-sep-inner-induction}) to establish Lemma \ref{lem: well-sep-invariant}, we will need to couple the inner induction with the second part of the claim to actually prove (\ref{eq: well-sep-inner-induction}). The second part of the claim states that, upon the arrival of the single point $x_i$, as we consecutively reset the pivots until no Add Operation or Exchange Operation is available, the (fixed) estimated center $c_j^T$ for cluster $j$ at the end of phase $T$ remains attached to the pivot for label $j$ --- even though the well-separation parameter $\beta_t$ decreases and the pivot for label $j$ may change position.

\begin{proof}[Proof of Proposition \ref{prop: intermediate_attachment}]
The proof is by induction. First we prove the base case $t = T$. We know that $p_1^T, \dots, p_T^T$ are $\beta_{T}$-well-separated w.r.t. the weights at the start of phase $T$ by the inductive assumption of the outer induction, thus also w.r.t. $w_{i-1}$. So (\ref{eq: well-sep-inner-induction}) is satisfied. Also, we know that $c_j^T$ is $\beta_{T+1}$-attached to $p_j^T$ w.r.t. $w_{i-1}$ by Proposition \ref{prop: attached_estimated_center}, so (b) is satisfied. This concludes the base case.  

Next we assume the claim holds for $t \in [T, T^+)$. The inductive step is to prove that the claim holds for $t+1$, or when this phase is skipped as in Case 4 of the Add Operation and Case 5 of the Exchange Operation, holds for $t+2$.

\paragraph*{Inductive step for (\ref{eq: well-sep-inner-induction}).} We start by proving the inductive step for the first part (\ref{eq: well-sep-inner-induction}) of the claim for $t+1$ or $t+2$. If one of Cases 1 through 3 of the Add Operation is executed, then (\ref{eq: well-sep-inner-induction}) holds by construction for $t+1$. This is because in these cases the new pivot location ($x_{\alpha}$ or some $c_j^T$, $j \in [T]$) is $\beta_{t+1}$-well-separated from $p_1^t, \cdots, p_t^t$ w.r.t. $w_i = w_{t+1}$. Likewise, if one of Cases 1 through 4 of the Exchange Operation is executed, then (\ref{eq: well-sep-inner-induction}) holds by construction for $t+1$. This is because in these cases the new set of pivots is $\{p_1^t, \cdots, p_t^t\} \cup \{x_{\alpha},x_{\gamma}\} \setminus \{p_{j}^t\}$, which is $\beta_{t+1}$-well-separated w.r.t. $w_i = w_{t+1}$. So it remains to show that (\ref{eq: well-sep-inner-induction}) holds for $t+2$ in Case 4 of the Add Operation and Case 5 of the Exchange Operation.

\setcounter{case}{0}

\begin{case} \label{cs: well-sep-invariant-3b}
Case 4 of the Add Operation is executed.
\end{case}
 In particular, we need to show that (A) $c_{f}^T$ and $c_{g}^T$ are $\beta_{t+2}$-well-separated w.r.t. $w_{t+2} = w_i$, and (B) $c_{f}^T$ and $c_{g}^T$ are each $\beta_{t+2}$-well-separated from $p_1^t, \dots, p_t^t$ w.r.t. $w_{t+2} = w_i$. \\
 
First we prove (A). Suppose to the contrary that $c_{f}^T$ and $c_{g}^T$ are $\beta_{t+2}$-attached w.r.t. $w_i$. Then $c_{f}^T$ and $c_{g}^T$ are $\beta_{t+2}$-attached w.r.t. $w_{i-1}$, since $w_{i-1} \leq w_i$ holds pointwise. Also, by Proposition \ref{prop: attached_estimated_center}, $c_{f}^T$ is $\beta_{T+1}$-attached to $p_{f}^T$ w.r.t. $w_{i-1}$ and $w_{i-1}(c_{f}^T) \geq w_{i-1}(p_{f}^T)$ (and likewise for $c_g^T, p_g^T$). Since by the base case $p_1^T, \cdots, p_T^T$ are $\beta_{T}$-well-separated w.r.t. $w_{i-1}$, and $\beta_{t+2} + \beta_{T+1} + \beta_{T+1} < \beta_T$, we reach a contradiction through two applications of Proposition \ref{prop: meta-prop}. We conclude that $c_{f}^T$ and $c_{g}^T$ are $\beta_{t+2}$-well-separated w.r.t. $w_{t+2} = w_i$. \\

Next we prove (B). WLOG, we will first show that $c_f^T$ is $\beta_{t+2}$-well-separated from $p_f^t$ w.r.t. $w_{t+2} = w_i$. First, we show that $w_i(c_f^T) \geq w_i(p_f^t)$. This is clearly true for $t>T$, by assumption of Case 4 of the Add Operation and the fact that $w_t = w_i$. When $t=T$, the argument is more subtle. Note that $c_f^T \neq p_f^t$, since $c_f^T$ is $\beta_{t+2}$-attached to $x_{\alpha}$ w.r.t. $w_i$, while $x_{\alpha}$ is $\beta_{t+1}$-well-separated from $p_f^t$ w.r.t. $w_i$. In turn, since $c_f^T \neq p_f^t$ if $t = T$, by Proposition \ref{prop: attached_estimated_center} we have $w_{i-1}(c_f^T) > w_{i-1}(p_f^t)$. Since the natural weights can only go up by 1 on the arrival of $x_i$, this means that $w_i(c_f^T) \geq w_i(p_f^t)$ if $t=T$. So we conclude $w_i(c_f^T) \geq w_i(p_f^t)$ for all $t \in [T, T^+)$ as desired. Thus, by Proposition \ref{prop: meta-prop}, it is not possible for $c_{f}^T$ to be $\beta_{t+2}$-attached to $p_f^t$ w.r.t. $w_i$, since $w_i(c_f^T) \geq w_i(p_f^t)$, $c_f^T$ is $\beta_{t+2}$-attached to $x_{\alpha}$ w.r.t. $w_i$ (by definition of Case 4 of the Add Operation), and $p_f^t$ and $x_{\alpha}$ are $\beta_{t+1}$-well-separated w.rt. $w_i$ (by definition of the Add Operation). 

To finish the proof of (B), we need to show that $c_f^T$ is $\beta_{t+2}$-well-separated from $p_l^t$ w.r.t. $w_{t+2}$, $l \in [t] \setminus f$. By the inductive assumption, at least one of statements (a)---(c) in the statement of the proposition holds. Note that (a) in the statement of the proposition does not hold, for else Case 1 of the Add Operation would have been executed. Also (c) in the statement of the proposition does not hold, since by assumption of Case 4 of the Add Operation, $w_t(c_f^T) \geq w_t(p_f^t)$. So (b) in the statement of the claim holds, i.e., $c_f^T$ is $\beta_{t+1}$-attached to $p_f^t$ w.r.t. $w_t$. Also by the inductive assumption, $p_f^t$ and $p_l^t$ are $\beta_t$-well-separated w.r.t. $w_t$. Since $\beta_{t+1} + \beta_{t+2} < \beta_t$, it must be the case by Proposition \ref{prop: meta-prop} that $c_f^T$ and $p_l^t$ are $\beta_{t+2}$-well-separated w.r.t. $w_t$ for $l \neq f$. Thus $c_f^T$ and $p_l^t$ are $\beta_{t+2}$-well-separated w.r.t. $w_{t+2}$ as well, concluding the proof of (B) and of Case \ref{cs: well-sep-invariant-3b}.

\begin{case}
Case 5 of the Exchange Operation is executed.
\end{case}
First note that $\{x_{\alpha},x_{\gamma}\} \cup \{p_l^t \mid l \in [t] \setminus j\}$ is $\beta_{t+1}$-well-separated w.r.t. $w_{t+2} = w_i$ by assumption of the Exchange Operation, thus this set is also $\beta_{t+2}$-well-separated w.r.t. $w_i$. 

Next, we need to show that $c_{j}^T$ is $\beta_{t+2}$-well-separated from $\{p_l^t \mid l \in [t] \setminus j\} \cup \{x_{\alpha},x_{\gamma}\}$ w.r.t. $w_i$. Since Cases 3 and 4 of the Exchange Operation were not executed, $c_{j}^T$ is $\beta_{t+2}$-well-separated from $x_{\alpha}$ and from $x_{\gamma}$ w.r.t. $w_i$. 

Next, note that since an Add Operation, which takes precedence over an Exchange Operation, was not executed, (a) in the statement of the proposition cannot hold in the inductive assumption for $j$. Moreover, since Case 2 of the Exchange Operation is not executed and (c) in the statement of the proposition cannot hold when $t=T$ (by Proposition \ref{prop: attached_estimated_center}), (c) in the statement of the proposition cannot hold in the inductive assumption for $j$. We conclude $c_j^T$ is $\beta_{t+1}$-attached to $p_{j}^t$ w.r.t. $w_t$ and $w_t(c_{j}^T) \geq w_t(p_{j}^t)$. This implies by Proposition \ref{prop: meta-prop} that $c_{j}^T$ is $\beta_{t+2}$-well-separated from $p_l^t$ w.r.t. $w_t$ for $l \neq j$, since $p_{j}^t$ and $p_l^t$ are $\beta_{t}$-well-separated w.r.t. $w_t$ by the inductive assumption. Hence $c_{j}^T$ is $\beta_{t+2}$-well-separated from $p_l^t$ w.r.t. $w_{t+1} = w_i$ for $l \in [t] \setminus j$.

\paragraph*{Inductive step for (a)---(c).}  Next we prove the inductive step holding for the second part of the proposition. We first introduce the following fact which we will use repeatedly in the rest of the proof.

 \begin{fact} \label{fact: centers-well-sep}
Let $j, j' \in [T]$, $j \neq j'$. Then $c_j^T$ and $c_{j'}^T$ are $\beta_{T+1}$-well-separated w.r.t. $w_{i-1}$.
\end{fact}

\begin{proof}[Proof of Fact \ref{fact: centers-well-sep}]
Suppose to the contrary that $c_j^T$ and $c_{j'}^T$ are $\beta_{T+1}$-attached w.r.t. $w_{i-1}$. We also know by Proposition \ref{prop: attached_estimated_center} that $c_j^T, c_{j'}^{T}$ are $\beta_{T+1}$-attached to $p_j^T, p_{j'}^T$, respectively, w.r.t. $w_{i-1}$, and that $w_{i-1}(c_j^T) \geq w_{i-1}(p_j^T)$ and $w_{i-1}(c_{j'}^T) \geq w_{i-1}(p_{j'}^T)$. By the base case of the inner induction, $p_1^T, \cdots, p_T^T$ are $\beta_T$-well-separated w.r.t. $w_{i-1}$, so applying Proposition \ref{prop: meta-prop} gives a contradiction. 
\end{proof}

In the remainder of the proof, we need to show that, given the inductive assumption that Proposition \ref{prop: intermediate_attachment} holds for $t$, we have that at least one of statements (a)--(c) holds for $t+1$ (or when this phase is skipped, $t+2$). (Recall we already proved the first part of the inductive step, i.e.,  (\ref{eq: well-sep-inner-induction}) holds for $t+1$.) 

We case on which operation is executed to generate the new set of $t+1$ (or $t+2$) pivots from $p_1^t, \dots, p_t^t$, and on which of statements (a)--(c) holds in the inductive assumption for $t$.

\setcounter{case}{0}

\begin{case} \label{cs1addmove}
Case 1 of the Add Operation is executed.
\end{case}
Let $j^*$ be the specific $j$ chosen in Case 1 of the Add Operation. We need to show that for every $j \in [T]$, at least one of statements (a)--(c) holds for $t+1$. \\
If $j = j^*$, then $p_j^{t+1} = c_j^T$, so certainly $c_j^T$ is $\beta_{t+2}$-attached to $p_j^{t+1}$ w.r.t. $w_{t+1}$, so statement (b) is satisfied for $t+1$. \\
Now we consider $j \neq j^*$. We case on which of statements (a)--(c) holds for $j \in [T]$ in the inductive assumption for $t$. \\

\noindent \underline{If statement (a) holds in the inductive assumption for $t$.} 
Then by definition, $c_j^T$ is $\beta_{t+1}$-well-separated from $p_1^t, \dots, p_t^t$ w.r.t. $w_i$, so it must also be $\beta_{t+2}$-well-separated from $p_1^t, \dots, p_t^t$ w.r.t. $w_i$. Moreover, by Fact \ref{fact: centers-well-sep}, $c_j^T$ is $\beta_{T+1}$-well-separated (thus $\beta_{t+2}$-well-separated) from $p_{j^*}^{t+1} = c_{j^*}^T$ w.r.t. $w_{i-1}$, thus also w.r.t. $w_i$. Since $\{p_1^{t+1}, \dots, p_{t+1}^{t+1}\} = \{p_1^t, \dots, p_t^t, c_{j^*}^T\}$, statement (a) holds for $c_j^T$ and the inductive step $t+1$ is satisfied. 
\\

\noindent \underline{Else if statement (c) holds in the inductive assumption for $t$.} 
Note that $p_j^t = p_j^{t+1}$ for $j \neq j^*$ when Case 1 is executed. Moreover, if statement (c) holds, it must be the case that $t > T$, since $w_{i-1}(c_j^T) \geq w_{i-1}(p_j^T)$ by construction. But if $t > T$, then $w_{t+1} = w_{t} = w_i$. Also, $f(t+1, T) > f(t,T)$. So statement (c) still holds for $t+1$. 
\\

\noindent \underline{Else statement (b) holds in the inductive assumption for $t$.} 
In this case, since statement (c) does not hold, then $c_j^T$ is $\beta_{t+1}$-attached to $p_j^t$ w.r.t. $w_t$, and $w_t(c_j^T) \geq w_t(p_j^t)$ (the latter inequality holds because if $t=T$, this is true by definition, and if $t >T$, then $f(t,T) = \beta_T(t-T) > \beta_{t+1}$.) It is not possible for $c_j^T$ to be $\beta_{t+2}$-attached to $p_{j'}^t$ w.r.t. $w_{t+1}$ for $j'\in [t] \setminus j$ , since this would mean $c_j^T$ is $\beta_{t+2}$-attached to $p_{j'}^t$ w.r.t. $w_t$ as well. But if this were true, then we get a contradiction by Proposition \ref{prop: meta-prop}, because $p_j^t$ and $p_{j'}^t$ are $\beta_t$-well-separated w.r.t. $w_t$ by the inductive assumption. Also, it is not possible for $c_j^T$ to be $\beta_{t+2}$-attached to $p_{j^*}^{t+1} = c_{j^*}^T$ w.r.t. $w_{t+1}$, thus also w.r.t. $w_{i-1}$, by Fact \ref{fact: centers-well-sep}. So either $c_j^T$ is $\beta_{t+2}$-well-separated from $\{p_1^t, \dots, p_t^t, c_{j^*}^T\} = \{p_1^{t+1}, \dots, p_{t+1}^{t+1} \}$ w.r.t. $w_{t+1} = w_i$; or, $c_j^T$ is $\beta_{t+2}$-attached to $p_j^t = p_j^{t+1}$ w.r.t. $w_{t+1}$. So statement (a) or statement (b) must hold for $t+1$. 
\\

 In the remaining cases of the Add Operation, it is not possible that statement (a) holds in the inductive assumption for $t$, for then Case 1 of the Add Operation would have been executed. Thus we can omit statement (a) from the casework.

\begin{case} \label{cs2addmove}
Case 2 of the Add Operation is executed. 
\end{case}

\noindent \underline{If statement (c) holds in the inductive assumption for $t$.} The reasoning is the same as in Case \ref{cs1addmove} (except there is no distinguished $j^*$). \\

\noindent \underline{Else statement (b) holds in the inductive assumption for $t$.} Since statement (c) does not hold, we have $w_t(c_j^T) \geq w_t(p_j^t)$ (the reasoning is the same as in Case \ref{cs1addmove}). But by the assumptions of Case 2 of the Add Operation, this means that $c_j^T$ is $\beta_{t+2}$-well-separated from $x_{\alpha}$ w.r.t. $w_i$.  So either $c_j^T$ is $\beta_{t+2}$-well-separated from $\{p_1^t, \dots, p_t^t, x_{\alpha}\} = \{p_1^{t+1}, \dots, p_{t+1}^{t+1}\}$ w.r.t. $w_i$, or $c_j^T$ is $\beta_{t+2}$-attached to $p_j^t = p_j^{t+1}$ w.r.t. $w_i = w_{t+1}$, by the same reasoning as in Case \ref{cs1addmove} (except there is no distinguished $j^*$). So statement (a) or statement (b) must hold for $t+1$.

\begin{case} \label{cs3a_addmove}
Case 3 of the Add Operation is executed. 
\end{case}

Let $j^*$ be as in $j$ chosen in Case 3 of the Add Operation. We need to show that for every $j \in [T]$, at least one of statements (a)---(c) holds for $t+1$. \\
If $j=j^*$, then $p_j^{t+1} = x_{\alpha}$, and by definition of $j^*$, $c_j^T$ is $\beta_{t+2}$-attached to $p_j^{t+1}$ w.r.t. $w_i = w_{t+1}$, so statement (b) is satisfied for $t+1$.\\
Now we consider $j \neq j^*$. We case on which of statements (a)---(c) holds for $j \in [T]$ in the inductive assumption for $t$.  \\

\noindent \underline{If statement (c) holds in the inductive assumption for $t$.} The reasoning is the same as in Case \ref{cs1addmove}. \\

\noindent \underline{Else statement (b) holds in the inductive assumption for $t$.} Note that since $w_t(c_j^t) \geq w_t(p_j^t)$ and $j \neq j^*$  (for the same reasoning as in Case \ref{cs1addmove}),  $c_j^T$ is $\beta_{t+2}$-well-separated from $p_{j^*}^{t+1} = x_{\alpha}$ w.r.t. $w_i$ (by the assumptions of Case 3 of the Add Operation and the uniqueness of $j^*$). Now the reasoning is the same as in Case \ref{cs1addmove}.

\begin{case} \label{cs3b_addmove}
Case 4 of the Add Operation is executed. 
\end{case}

Let $f$ or $g$ be as in Case 4 of the Add Operation. We need to show that for every $j \in [T]$, at least one of the statements (a)---(c) holds for $t+2$. (Note phase $t+1$ is skipped.)
If $j = f$ or $g$, then $p_j^{t+2} = c_j^T$, so certainly $c_j^T$ is $\beta_{t+3}$-attached to $p_j^{t+2}$ w.r.t. $w_{t+2}$, so statement (b) is satisfied for $t+2$. 
Now we consider $j \neq f,g$. We case on which of statements (a)---(c) held for $j \in [T]$ in the inductive assumption for $t$. \\

\noindent \underline{If statement (c) holds in the inductive assumption for $t$.} The reasoning is nearly the same as in Case \ref{cs1addmove}, except now there are two distinguished $f,g$ (instead of just $j^*$), and we use that $f(t+2, T) > f(t,T)$ (instead of that  $f(t+1, T) > f(t,T)$).  \\

\noindent \underline{Else statement (b) holds in the inductive assumption for $t$.} Note that since $j \neq f,g$,  $c_j^T$ is $\beta_{t+3}$-well-separated from $p_{f}^{t+2} = c_{f}^T$ w.r.t. $w_i$ and from $p_{g}^{t+2} = c_{g}^T$ w.r.t. $w_i$, by Fact \ref{fact: centers-well-sep}. By the same reasoning as in Case \ref{cs1addmove}, either $c_j^T$ is $\beta_{t+3}$-attached to $p_j^{t+2} = p_j^t$ w.r.t. $w_i$, or $c_j^T$ is $\beta_{t+3}$-well-separated from $\{p_1^t, \dots, p_t^t, c_{f}^T, c_{g}^T\} = \{p_1^{t+2}, \dots, p_{t+2}^{t+2}\}$ w.r.t. $w_{i}$. So statement (a) or statement (b) must hold for $t+2$. \\

In all cases of the Exchange Operation, it is not possible that statement (a) holds in the inductive assumption for $t$, for then an Add Operation (specifically Case 1) would have been executed instead, as the Add Operation takes precedence over the Exchange Operation. Thus we can omit statement (a) from the casework. Let $j^*$ be the specific $j$ in the description of the Exchange Operation in the algorithm description.

\begin{case} \label{cs1_exchange_move}
Case 1 of the Exchange Operation is executed. 
\end{case}

We need to show that for every $j \in [T]$, at least one of statements (a)---(c) holds for $t+1$. \\
If $j=j^*$, there is nothing to prove because in Case 1 of the Exchange Operation, $j^* > T$, meaning $c_{j^*}^T$ is not defined.\\
Now we consider $j \neq j^*$. We case on which of statements (a)---(c) holds for $j \in [T]$ in the inductive assumption for $t$.  \\

\noindent \underline{If statement (c) holds in the inductive assumption for $t$.} The reasoning is the same as in Case \ref{cs1addmove}.  \\

\noindent \underline{Else statement (b) holds in the inductive assumption for $t$.} We first claim $c_j^T$ is $\beta_{t+1}$-attached to $p_j^t = p_j^{t+1}$ w.r.t. $w_{t+1}$ (note (b) states this w.r.t. $w_t$ only). For, given that $c_j^T$ is $\beta_{t+1}$-attached to $p_j^t$ w.r.t. $w_t$ and $w_t(c_j^T) \geq w_t(p_j^t)$, it must be the case that $c_j^T$ is $\beta_{t+1}$-well-separated from $p_{j'}^t$ w.r.t. $w_{t}$, for any $j' \neq j$ (by Proposition \ref{prop: meta-prop} and the fact that $p_j^t$ and $p_{j'}^t$ are $\beta_t$-well-separated w.r.t. $w_t$). But this means that $c_j^T$ is $\beta_{t+1}$-well-separated from $p_{j'}^t$ w.r.t. $w_{t+1}$ as well. Thus, if $c_j^T$ were not $\beta_{t+1}$-attached to $p_j^t$ w.r.t. $w_{t+1}$, then it would be $\beta_{t+1}$-well-separated from $\{p_1^t, \dots, p_t^t\}$ w.r.t. $w_{t+1} = w_i$, hence an Add Operation would have been executed instead as the Add Operation  takes precedence over the Exchange Operation. So we conclude that $c_j^T$ is $\beta_{t+1}$-attached to $p_j^t = p_j^{t+1}$ w.r.t. $w_{t+1}$.
\\
If $w_{t+1}(c_j^T) < w_{t+1}(p_j^t)$, then since $t+1 > T$,  statement (c) holds for $t+1$. Otherwise, $w_{t+1}(c_j^T) \geq w_{t+1}(p_j^t)$. Also, $w_{t+1}(x_{\alpha}) \geq w_{t+1}(p_{j^*}^t)$ and $x_{\alpha}$ is $\beta_{t+1}$-attached to $p_{j^*}^t$ w.r.t. $w_{t+1} = w_i$, by definition of the Exchange Operation. So noting that $p_j^t$ and $p_{j^*}^t$ are $\beta_t$-well-separated w.r.t. $w_{t+1}$ and applying Proposition \ref{prop: meta-prop} twice gives that $c_j^T$ is $\beta_{t+2}$-well-separated from $x_{\alpha}$ w.r.t. $w_{t+1}$ (and likewise for $x_{\gamma}$). We also know that $c_j^T$ is $\beta_{t+2}$-well-separated from $p_{j'}^t$ w.r.t. $w_{t+1}$, $j' \neq j$, by the previous paragraph. So either $c_j^T$ is $\beta_{t+2}$-attached to $p_j^t = p_j^{t+1}$ w.r.t. $w_{t+1}$, or $c_j^T$ is $\beta_{t+2}$-well-separated from $\{p_1^t, \cdots, p_t^t, x_{\alpha}, x_{\gamma}\} \setminus \{p_{j^*}^t\} = \{p_1^{t+1}, \dots, p_{t+1}^{t+1}\}$. Thus statement (a) or (b) holds for $t+1$.

\begin{case} \label{cs2a_exchange_move}
Case 2 of the Exchange Operation is executed.
\end{case}

We need to show that for every $j \in [T]$, at least one of statements (a)---(c) holds for $t+1$. \\
First consider $j=j^*$. Since as noted above statement (a) cannot hold for the inductive assumption for $t$, statement (b) or statement (c) must hold for the inductive assumption for $t$. Also, since $w_{t+1} = w_i$, it must be the case that $w_{t+1}(c_j^T) < w_{t+1}(p_j^t)$ by the assumptions of Case 2 of the Exchange Operation. Thus $w_t(c_j^T) \leq w_t(p_j^t)$. 

If $t = T$, then it must be the case that $w_t(c_j^T) = w_t(p_j^t)$ hence $c_j^T = p_j^t$ by the tiebreaking rule in the definition (\ref{eq: good_center}) of estimated center. If $t > T$, then since $w_t = w_{t+1} = w_p$, statement (c) holds for $t$. So in both cases, $w_i(c_j^T) \cdot d(c_j^T, p_j^t) \leq \beta_T(t-T)$. Moreover, by the definition of the Exchange Operation, $w_i(p_j^t) \cdot d(p_j^t, p_j^{t+1}) = w_i(p_j^t) \cdot d(p_j^t, x_{\alpha}) \leq \beta_{t+1} \cdot \textsf{OPT}$. So by Proposition \ref{prop: meta-prop}, $w_i(c_j^T) \cdot d(c_j^T, x_{\alpha}) \leq (\beta_{t+1} + f(t,T)) \textsf{OPT} \leq f(t+1, T) \textsf{OPT}$. Hence statement (c) holds for $t+1$.  

 For $j \neq j^*$, the reasoning is the same as in Case \ref{cs1_exchange_move}.

\begin{case} \label{cs2b_exchange_move}
Case 3 of the Exchange Operation is executed. 
\end{case}

We need to show that for every $j \in [T]$, at least one of statements (a)---(c) holds for $t+1$. \\
First consider $j=j^*$. By the assumptions of of Case 3 of the Exchange Operation, statement (b) holds for $t+1$.  

For $j \neq j^*$, the reasoning is the same as in Case \ref{cs1_exchange_move}.

\begin{case} \label{cs2c_exchange_move}
Case 4 of the Exchange Operation is executed. 
\end{case}

The reasoning is the same as in Case \ref{cs2b_exchange_move}, only now the roles of $x_{\alpha},x_{\gamma}$ are reversed.

\begin{case} \label{cs2d_exchange_move}
Case 5 of the Exchange Operation is executed. 
\end{case}

We need to show that for every $j \in [T]$, at least one of statements (1)--(3) holds for $t+2$. (Note step $t+1$ is skipped.) \\
First consider $j = j^*$. $c_{j^*}^T$ is $\beta_{t+3}$-attached to $p_{j^*}^{t+2}= c_{j^*}^T$ w.r.t. $w_{t+2}$, so statement (b) holds for for $t+2$. \\
For $j \neq j^*$, $c_j^T$ is $\beta_{t+3}$-well-separated from $c_{j^*}^T$ w.r.t. $w_i$ (Fact \ref{fact: centers-well-sep}). The reasoning is now the same as in Case \ref{cs1_exchange_move} once we note that $\beta_{t+2} > \beta_{t+3}$. 
\end{proof}
Having proven Proposition \ref{prop: intermediate_attachment}, the proof of Lemma \ref{lem: well-sep-invariant} is now complete.  
\end{proof}

\section{Omitted Proofs from Section \ref{sec: proofs_bounded_cost}} \label{appendix: proofs_bounded_cost}

\begin{proof}[Proof of Lemma \ref{lem: sufficiently_weighted_centers}]
 As in Lemma \ref{lem: far_points}, let $C_j = C_j^T \setminus C_j^{T^-}$ and let $S_{ji}$ be the set of elements in $C_j$ assigned to $y_i$ in the clustering of $X(T) \setminus X(T^-)$ induced by $P_T$. Let $S_{far, j} = \bigcup_{i : p(y_i) \neq p_j^T} S_{ji}$, $S_{near, j} = \bigcup_{i : p(y_i) = p_j^T} S_{ji}$, and $S_j$ be the elements in $C_j$ that are assigned to $p_j^T$ in the clustering of $X(T) \setminus X(T^-)$ induced by $P_T$. 

The proof is by induction. We have that 
\begin{equation} \label{eq: set_szs}
|C_j^T| = |C_j^{T^-}| + |C_j| = |C_j^{T^-}| + |S_{far,j}| + |S_{near, j}| + |S_j|
\end{equation}
(Note we use that there are no points in $C_j$ that are assigned to $p_{j'}^T$, $j' \neq j$, in the offline clustering induced by $P_T$, due to the greedy labelling rule. This is true as long as in the offline clustering induced by $P_T$ we break ties consistent with how the online algorithm breaks ties.) 

 First, we bound the last three terms. Let $w^t$ denote the natural weights at the end of phase $t$. 
\begin{equation} \label{eq: size_far}
|S_{far, j}| \leq k \cdot w^T(p_j^T) \leq k \cdot w^T(c_j^T)
\end{equation}
where the first inequality follows from Lemma \ref{lem: far_points} and the second inequality follows from the definition (\ref{eq: good_center}) of estimated center. Next, 
\begin{equation} \label{eq: sz_near}
|S_{near, j}| = \sum_{i:p(y_i) = p_j^T} |S_{ji}| \leq \sum_{i:p(y_i) = p_j^T} w^T(y_i) \leq k \cdot w^T(c_j^T)
\end{equation}
where the second inequality follows from the definition of $w^T$. The third inequality follows from the definitions of attachment digraph and estimated center: If $y_i \in \delta^-(p_j^T)$, then $w^T(y_i) \leq w^T(p_j^T)$ by construction of the attachment digraph $D(T)$. Otherwise, $y_i \in \delta^+(p_j^T)$, so by (\ref{eq: good_center}), $w^T(c_j^T) \geq w^T(y_i)$. Finally, 
\begin{equation} \label{eq: size_j}
|S_j| \leq w^T(p_j^T) \leq w^T(c_j^T)
\end{equation}
where the first inequality is by the definition of $w^T$ and the second inequality from (\ref{eq: good_center}).

For simplicity, let $h(t,k) = (2k+1)t$. 
Now we need to bound $|C_j^{T-}|$ in terms of $w^T(c_j^T)$. If $j \not \in [T^-]$, then $|C_j^{T^-}| = 0$. So assume $j \in [T^-]$. Inductively, we have that 
$$|C_j^{T^-}| \leq h(T^-,k) \cdot w^{T^-}(c_j^{T^-}).$$
We will prove that 
\begin{equation} \label{eq: sz_prev}
|C_j^{T^-}| \leq h(T^-, k) \cdot w^T(c_j^T).
\end{equation}

There are two cases to consider. 

\setcounter{case}{0}

\begin{case}
(a) holds in Lemma \ref{lem: well_attached_centers}.
\end{case}
$$|C_j^{T^-}| \leq h(T^-,k) \cdot w^{T^-}(c_j^{T^-}) \leq h(T^-,k) \cdot w^T(p_j^T) \leq h(T^-,k) \cdot w^T(c_j^T).$$

\begin{case}
(b) holds in Lemma \ref{lem: well_attached_centers}.
\end{case}

This means that $c_j^{T^-}$ is $\beta_{T+1}$-attached to $p_j^T$ w.r.t. $w^T$. If $w^T(c_j^{T^-}) 
 \leq w^T(p_j^T)$, then $w^T(c_j^{T^-}) \leq w^T(c_j^T)$. Otherwise, $w^T(c_j^{T^-}) > w^T(p_j^T)$, so $c_j^{T^-} \in \delta^+(p_j^T)$. By (\ref{eq: good_center}), $w^T(c_j^{T^-}) \leq w^T(c_j^{T})$. In both cases we have $w^T(c_j^{T^-}) \leq w^T(c_j^T)$, so building from the inductive assumption,
 $$|C_j^{T^-}| \leq h(T^-,k) \cdot w^{T^-}(c_j^{T^-}) \leq h(T^-, k) \cdot w^T(c_j^{T^-}) \leq h(T^-, k) \cdot w^T(c_j^T)$$
 which concludes the case. Putting equations (\ref{eq: set_szs}), (\ref{eq: size_far}), (\ref{eq: sz_near}), (\ref{eq: size_j}), (\ref{eq: sz_prev}) together gives
 $$|C_j^T| \leq \left(h(T^-,k) + 2k+1 \right) \cdot w^T(c_j^T) \leq h(T,k) \cdot w^T(c_j^T) = (2k+1)\cdot T \cdot w^T(c_j^T)$$
as desired. 

\end{proof}

\section{Decreasing the well-separation parameter is necessary.} \label{appendix: decreasing_well_sep}
Due to the relatively static nature of our algorithm (we only reset the pivots at most $k$ times), it is perhaps not surprising that we will have to initialize the well-separation parameter $\beta_1$ to be large, and then decrease it as we create more pivots. As a result, our algorithm will accumulate $\Omega(poly(k) 2^k)$ cost. Next, we show decreasing $\beta$ is necessary to avoid accumulating infinite cost.  

In fact, we require $\beta_{t+1} \leq \beta_t/2$. See Figure \ref{fig: halving_beta_ex}. Let the phase be $t=2$ and $p_1^2$ and $p_2^2$ denote the pivots in phase 2, associated with labels 1 and 2, respectively. Let $\beta = \beta_2$ and $\eps >0$. 

By the time phase 2 starts, the only points that have arrived are one point each at $p_1^2$ and $p_2^2$. Note that since $p_1^2$ and $p_2^2$ are distance $\beta$ apart, they satisfy the invariant that the pivots at the start of phase 2 should be $\beta = \beta_2$-well-separated w.r.t. the natural weights at the start of phase 2. Suppose we do not require $\beta_3 \leq \beta_2/2$. If this is the case, we may have infinitely many points arrive at $x$ and $y$, without ever being able to make three pivots. Moreover, since we assign points to pivots greedily, the infinitely many points at $x$ and $y$ would be assigned the same label, yielding a cluster with unbounded cost. To see that three pivots could never be made, note that $p_1^2$ prevents us from adding $x$ as a pivot, since $p_1^2$ and $x$ are $(\beta/2 + \eps)$-attached at any point in time (even though there are infinitely many points at $x$). Likewise, $p_2^2$ prevents us from adding $y$ as a pivot. So as long as at least one of $p_1^2$ or $p_2^2$ must remain a pivot (which must be the case if we are to increase the number of pivots to three, as there are only four locations in this example), it is impossible to have $\beta_3 > \beta_2/2$ without accumulating infinite cost. 

\begin{figure}[H]
\centering
\includegraphics[width=0.8\textwidth]{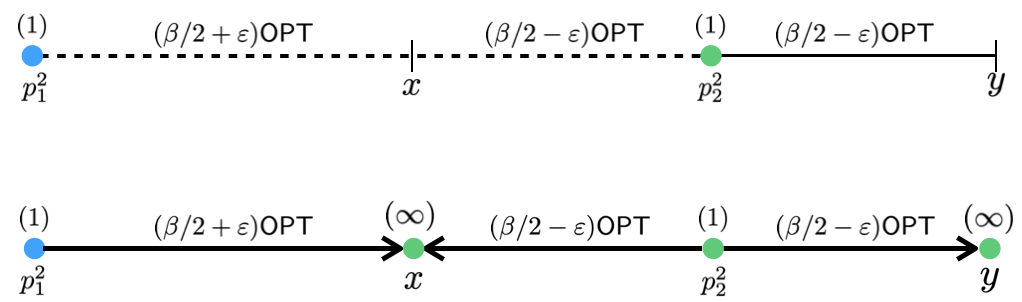}
\captionsetup{width=.9\linewidth}
\caption{The case discussed in Appendix \ref{appendix: decreasing_well_sep} showing that the well-separation parameter $\beta_t$ for phase $t$ must decrease by a factor of at least 2 with $t$.}
\label{fig: halving_beta_ex}
\end{figure}

\section{Avoiding label conflicts.} \label{appendix: label_conflicts}
Each pivot must be associated with a distinct label. The situation we have to handle delicately is when we are resetting pivots and two center estimates for two different online clusters are each close to the candidate new pivot; however, the candidate new pivot can only be associated with one label. See Figure \ref{fig: avoiding_label_conflicts_ex}. There are three present pivots, $p_1^2, p_2^3, p_3^3$, and the candidate new pivot is $x$, which is $\beta_4$-well-separated from the pivots $p_1^3$ and $p_2^3$ for labels 1 (blue) and 2 (green), respectively. However, $x$ is also ``close'' ($\beta/5$-attached, specifically) to the estimated centers $c_1^3$ and $c_2^3$ for clusters 1 and 2, respectively, so if we were to add $x$ as a pivot, it would be unclear as to whether we should give it the label 1 or 2. Instead, we will not actually add the candidate new pivot; we can show that the two center estimates $c_1^3$ and $c_2^3$ can be added as the pivots for the labels 1 and 2, respectively, while ensuring that all five pivots are $\beta/5$-well-separated. The old locations for pivots $p_1^3$ and $p_2^3$ are now the locations for the pivots of the new labels 4 (red) and 5 (purple). This situation is encoded in Case 4 of the Add Operation in Section \ref{sec: add_op}. 

 \begin{figure}[H]
\centering
\includegraphics[width=.9\textwidth]{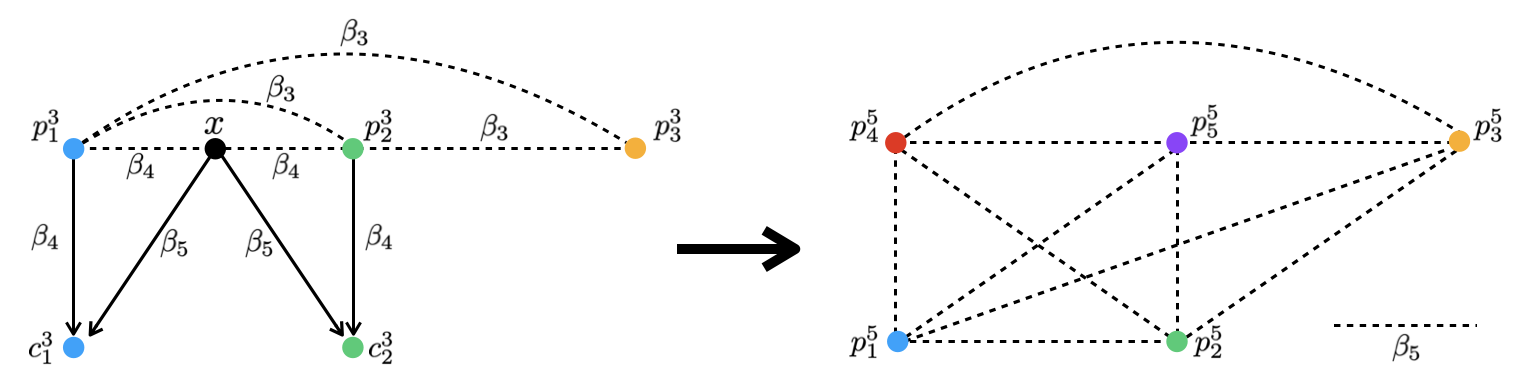}
\captionsetup{width=.9\linewidth}
\caption{The situation described in Appendix \ref{appendix: label_conflicts} in which we must add two pivots in one shot in order to avoid label conflicts. Dashed lines indicate well-separation and solid lines indicate attachment, labelled with the appropriate parameters. Arrows go from smaller to larger natural weights.}
\label{fig: avoiding_label_conflicts_ex}
\end{figure}

\end{document}